\documentclass[11pt]{article}
\usepackage{paper}

\title{Online Prediction in Sub-linear Space}
\author{Binghui Peng \\ Columbia University \\ \texttt{bp2601@columbia.edu}
\and Fred Zhang \\ UC Berkeley\footnote{Part of work done while interning at Google.} \\ \texttt{z0@berkeley.edu}}
\date{}

\makeatletter

\newcommand{\TODO}[1]{\textcolor{red}{[TODO\@ifnotempty{#1}{: #1}]}}
\newcommand{\fred}[1]{\textcolor{purple}{[F\@ifnotempty{#1}{: #1}]}}
\makeatother

\begin{document}

\maketitle
\begin{abstract}
We  provide the first sub-linear space and sub-linear regret algorithm for online learning with expert advice (against an oblivious adversary), addressing an open question raised recently by Srinivas, Woodruff, Xu and Zhou (STOC 2022). We also demonstrate a separation between oblivious and (strong) adaptive adversaries by proving a linear memory lower bound of any sub-linear regret algorithm against an adaptive adversary.

Our algorithm is based on a novel pool selection procedure that bypasses the traditional wisdom of leader selection for online learning, and a generic reduction that transforms any weakly sub-linear regret $o(T)$ algorithm to $T^{1-\alpha}$ regret algorithm, which may be of independent interest. Our lower bound utilizes the connection of no-regret learning and equilibrium computation in zero-sum games, leading to a proof of a strong lower bound against an adaptive adversary.
\end{abstract}

\newpage
\section{Introduction}\label{sec:intro}
Online prediction with expert advice is a fundamental task in sequential decision making and is the backbone of   optimization \cite{hazan2016introduction},  bandit learning \cite{lattimore2020bandit}, control theory \cite{doyle2013feedback}, among many other fields.
The problem is usually formulated as an online forecasting process that repeats for $T$ days.
 On each day, the algorithm is asked to provide a prediction of an outcome, given the advice from $n$ experts on the current day and all the previous information. 
After announcing a decision, it then receives feedback on the outcome and the loss of the expert predictions,  normalized to $[0, 1]^{n}$.
The objective in online learning is typically to compete against the best expert in hindsight, and the performance of an algorithm is measured in terms of \textit{regret}, which is defined as the additional cost of the algorithm over the best expert.

The celebrated Multiplicative Weights Update (MWU) method solves this online prediction problem with an optimal $O\left(\sqrt{T\log n}\right)$ regret \cite{arora2012multiplicative}.
Similar regret bounds have also been attained by weighted majority vote \cite{littlestone1989weighted}, follow-the-perturbed-leader (FTPL) scheme \cite{kalai2005efficient} and online mirror descent (OMD) \cite{hazan2016introduction}.
Besides the original motivation of online prediction, MWU and the follow-up variants find broader applications in algorithmic design, game theory and machine learning, with notable examples including efficient algorithms for linear programming and semi-definite programming \cite{garber2016sublinear}, approximate solution for max flow \cite{christiano2011electrical}, equilibrium computation \cite{cesa2006prediction, freund1999adaptive} and boosting \cite{freund1997decision}.


Despite a long history of research since 80's, the question of {\em space complexity} has been little explored in the online learning literature. All existing approaches explicitly track the cumulative cost of every expert and follow the advice of a (regularized or perturbed) leading expert, thus requiring a   memory of size $\Omega(n)$.
Motivated by this lack of understanding, a very recent work by Srinivas, Woodruff, Xu and Zhou \cite{stoc22} initiates the study of memory complexity of expert learning. 
They prove a lower bound of  ${\Omega}(\sqrt{nT/S})$ regret for any algorithm using $O(S)$ memory, even when the loss sequence are i.i.d distributed.
This implies that $\Omega(n)$ space is necessary to get the {\em optimal} $\sqrt{T}$ regret bound.
On the positive side, they   show that sub-linear regret is achievable in sub-linear space,   but only in random-order streams or when the best expert itself incurs sub-linear loss.
These structural assumptions  are arguably rather strong. 
This leaves open the  general question of attaining sub-linear regret in a memory-efficient fashion.

In this work, we revisit the classic online learning problem with an eye towards an better understanding its space complexity. 
On the algorithmic front, we ask, for   constant $\alpha, \beta \in (0, 1)$:
\begin{center}
  \textit{  Can we achieve sub-linear $O(T^{1-\alpha})$ regret  in online learning, using sub-linear $O(n^{\beta})$ space?}
\end{center}
We resolve this open question by designing the first sub-linear space online learning algorithm, in general worst-case settings.
The algorithm assumes an {\em oblivious adversary} that fixes the loss sequence in advance. 
To complement this,  we prove a lower bound showing that against an (strong) adaptive adversary, no sub-linear space algorithm can achieve sub-linear regret.

\subsection{Our results} 

To   introduce the results, we specify our problem setting (see \cref{sec:prelim} for a formal description). 
We consider a general setup of the expert problem  with $T$ days and $n$ experts. 
On each day $t \in [T]$, the algorithm decides to play one of the experts $i_t \in [n]$. Subsequently, nature reveals the loss $\ell_t(i)$ for all $i\in [n]$ and the algorithm incurs a loss of $\ell_t(i_t)$.  We assume that $\ell_t(i) \in [0,1]$ for all $t\in [T]$ and $i\in[n]$.  
The goal is to design an algorithm such that the (total) regret 
$\textsc{Regret}(T) = \E[\sum_{t\in [T]}\ell_t(i_t)] - \min_{i^{*} \in [n]} \sum_{t\in [T]}\ell_t(i^{*})$
is sub-linear in $T$.\footnote{We call $T^{1-\alpha}$ sub-linear regret (for constant $\alpha$), and $o(T)$ weakly sub-linear.  We use $\widetilde O(\cdot)$ to hide polylogarithmic factors in $n, T$ and $O_{n}(\cdot)$ to hide polynomial dependence on $n$. High probability refers to probability at least $1-1/\poly(T)$.}
An oblivious adversary  (randomly) chooses the loss vectors independent of the algorithm's decisions. 


Our main result is the following.

\begin{theorem}[Informal, see \autoref{thm:sublinear}]
\label{thm:intro-main}
Let $n, T \geq 1$, $\delta \in (0, 1]$. There exists an online learning algorithm that achieves a total regret of $O_n\left( T^{\frac{2}{2 +\delta}}\right)$ with high probability against an oblivious adversary, using $\widetilde O(n^{\delta} )$ memory.
\end{theorem}

The theorem provides a general memory-regret trade-off. To give some concrete examples:
\begin{itemize}
    \item We can get $ O_n\left(T^{4/5}\right)$   regret in $\widetilde O\left(\sqrt{n}\right)$ memory (by setting $\delta= 0.5$).  
    \item We can get $O_n (T^{0.67})$ regret in $\widetilde O\left({n}^{0.99}\right)$ memory (by setting $\delta= 0.99$).
\end{itemize}

On a conceptual level%
, our algorithm breaks the following traditional wisdom in online prediction. 
We recall that the classic regret-minimization algorithms, such as MWU and FTPL, all track the loss of {\em all} experts and follow a (regularized or perturbed) leader.
 Alternatively, one may assume an oracle that outputs the leading expert (a.k.a.\ oracle-efficient online learning \cite{hazan2016computational,dudik2020oracle}).
Therefore, the task of identifying a leading expert is usually believed to be a necessary sub-routine of regret minimization.
Indeed, with $\Omega(n)$   memory, identifying a leader and maintaining its cumulative loss is trivial, while in principle online prediction is a much harder task. 
Perhaps surprisingly, in the sub-linear space regime, our results suggest the opposite. 
It is   known that under space constraint, one cannot identify a leader, or even constant approximate its loss \cite{stoc22}. 
On the other hand, \cref{thm:intro-main} implies that achieving sub-linear regret is possible with arbitrarily small   polynomial  space. 
Therefore, conceptually our work shows that {\em online prediction is easier than and does not require  tracking   the leader, in the low space regime}.

In the case that only polylogarithmic space is allowed, we also give a weakly sub-linear regret algorithm. In particular, by setting $\eps = 1/\poly\log(nT)$ in the following theorem, the algorithm achieves $o(T)$ regret in $\widetilde O(1)$ memory (for sufficiently large $T$).
\begin{theorem}[Informal, see \cref{thm:weak-sublinear}]
\label{thm:intro-weak-sublinear}
Let $n \geq 1$, $\eps\in (1/\sqrt{n},1/2)$ and $T\geq \Omega(\eps^2 n)$. There exists an online learning algorithm that achieves a total regret of $\widetilde{O}\left(\eps T + T^{2/3}\left(\eps^2n\right)^{1/3}\right)$ with high probability against an oblivious adversary, using $\widetilde{O}(\eps^{-2})$ memory.
\end{theorem}


Our   algorithms assume   an oblivious adversary, a standard model in online learning as well as its cousin fields like online, streaming and dynamic algorithms \cite{karp1990optimal,alon1999space}.
In a variety of applications,  however, algorithms are required to work   against  adaptive adversaries  \cite{motwani1995randomized, slivkins2019introduction,wajc2020rounding, ben2022framework, alon2021adversarial}.

To complement our algorithmic results, we consider a strong adaptive adversary model, where the costs may be chosen adversarially that depend on the algorithm's prior randomness and decisions; see \cref{def:adaptive} for a formal definition.   
In this setting, MWU still achieves $\widetilde{O}(\sqrt{T})$ regret but uses  linear memory. 
We prove that $\Omega(n)$ memory is indeed necessary to obtain any {sub-linear} regret at all:
\begin{theorem}[Informal, see \cref{thm:lower}]
\label{thm:intro-lower}
Let $0 < \eps < 1/40$. Any algorithm that achieves $\eps T$ total regret against a strong adaptive adversary requires at least $\Omega\left(\min\left\{\eps^{-1}\log_2 n, n\right\}\right)$ bits of memory.
\end{theorem}

The theorem states that $\Omega(n)$ memory is necessary even to get any sub-linear, say, $O\left(T^{0.99}\right)$ regret, for sufficiently large $T = \Omega\left(n^{100}\right)$.
In contrast, under the oblivious adversary model, our upper bound (\cref{thm:intro-main}) can attain such  regret guarantee in $o(n)$ space.
Therefore, this exhibits a  separation between oblivious and adaptive adversary model in the low space regime.

\subsection{Technical overview} \label{sec:overview}
We provide a streamlined overview of our approach. For notational convenience, we omit polylogarithmic factors and one should think of $T = \poly(n)$.

\subsubsection{A baseline algorithm for weakly sub-linear regret}
We first present an algorithm that achieves weakly sub-linear regret $O(\eps T)$ using space $O(\eps^{-2})$, then provide a novel width reduction procedure to make it sub-linear in \cref{sec:boost}.

A   natural idea is to carefully maintain a small pool of experts and run MWU over them. 
To begin with, we divide the $T$ days into $T/B$ epochs with $B$ contiguous days in each epoch.
At the beginning of each epoch, we sample a random set of new experts and add them into the pool, then executing MWU over the pool (starting from uniform weight).
MWU guarantees that the algorithm is always competitive with the best expert(s) {\em in the pool}, and this turns online prediction into the task of maintaining a good pool of experts.
The immediate hurdle is that, due to space constraint, the pool   must be kept small and hence the best expert is likely to be outside of the pool.
Therefore, we need to maintain the pool so that it consists of good experts with respect to each particular epoch, but without knowing the performance of the experts outside the pool.

\paragraph{Maintaining the pool} 
The algorithm samples     a small number of experts into its pool at the beginning of each epoch. 
After a few epochs, the pool size would grow and exceed the memory budget.
The key algorithmic task is to design a rule of evicting  experts.

The first and natural idea is track the cumulative loss of each expert in the pool, and intuitively, only an expert with low average loss should be reserved. 
However, it is easy to come up with instances, where (1) the best expert $i^{*}$ has low error in every epoch; but (2) random new   experts have even lower error in $1/3$ epochs (though higher in the rest). 
If these $1/3$ epochs are evenly spaced, then $i^{*}$ can be easily kicked out by the new expert. 
Subsequently, it would take a long time to get back, since our sampling rate has to be low to respect the space constraint. This is clearly undesirable.

Intuitively, consider two experts in the pool with equally low average loss since their joining. The one that has stayed longer should be 
treated differently from the other, since the former is more stable against  further loss. Hence, the second idea is to keep experts that have lived a long time in the pool. In other words, an expert has no reason to be evicted if it is ``Pareto-optimal'': it has either stayed long or   achieved little loss.
Formally, we say an expert $j$ dominates another expert $i$, if $j$ joins the pool earlier than $i$ and   has achieved at most $\eps$ more average  loss than $i$ since expert $i$ joins.
At the end of every epoch, our algorithm evaluates all experts. Any expert $i$ being dominated by another is evicted. 

\paragraph{Bounding memory} 
 Notice that the algorithm does not dictate  any explicit bound on the pool size, let alone memory. This is the challenge that we now resolve.
 
For that, the key observation is  the following. Consider any surviving expert $i$ that joins the pool later than some other $j$. It  has much smaller, in fact at least $\eps$, loss  on its interval than $j$, by our eviction rule.  
Then, we claim that either (1) expert $j$ has $\eps/2$ larger average loss than expert $i$, or (2) $j$ lives $(1+\eps/2)$ longer than $i$. 
The reason is that     $j$ has loss at least $0$ everywhere, so if  (1) doesn't hold, then it takes extra length to catch up with the loss difference. 
A straightforward proof would bound the pool size by $O(1/\eps^2)$ as both events can happen for at most $O(1/\eps)$ times consecutively. 
We derive a   refined  one of    $O(1/\eps)$ via a potential function argument. 
This leads to a  memory bound of $O(1/\eps^2)$ because we need to track the performance of each expert  over every interval.

\paragraph{Bounding regret} To bound the regret, our plan is to show that there exist experts in pool that are competitive to $i^{*}$ (even if $i^{*}$ may not be in the pool), except for at most $O(n)$ {\em unlucky} epochs. 
For simplicity,   assume the algorithm only samples and adds one new expert per epoch. 
For an epoch $t$, imagine $i^{*}$ gets sampled, and it remains alive until the end of epoch $t \leq t'\leq T/B $. 
To this end, there must exist an expert $i(t)$ that already lies in the pool, stays alive during $[t, t']$ and outperforms $i^{*}$.
The reason is that the expert $i^{*}$ can only be evicted by an older expert, and the eviction time is independent of future randomness (again, we assume the loss sequence is fixed in advance). 
Hence, $i(t)$ is competitive with $i^{*}$  during the epochs $[t,t']$ (note this is independent of whether $i^{*}$ is actually sampled or not). Then we can proceed to $(t'+1)$-th epoch.
There is one exception: if $i^{*}$ is sampled and would not be evicted by the end, then $i(t)$ simply does not exist.
On one hand, with probability $1/n$, we would sample $i^{*}$ and it stays until the end. Otherwise, with probability $1-1/n$, we just lose this epoch and proceed to the next. The latter event should not happen for more than $O(n)$ times with high probability.

In summary, the baseline algorithm   is always competitive with best expert in pool, up to a total regret of $T/B \cdot O\left(\sqrt{B}\right)$, by MWU. Moreover, there is always some expert in the pool  competitive with  the best expert $i^{*}$  up to an $O(\eps T)$ regret, except in those unlucky epochs which incur $O(nB)$ regret.   To further optimize the algorithm, we  sample $1/\eps^2$ new experts instead of $1$. It turns out that allows us to bound the total regret over the unlucky epochs by $O(\eps^2 nB)$.  Putting everything together, the total regret is $O\left(\eps T + T/\sqrt{B} + \eps^2 n B\right) \approx O(\eps T)$.

\subsubsection{Bootstrap the baseline and width reduction}\label{sec:boost}
The above baseline algorithm has a total regret of $O(\eps T)$, and the bottleneck lies in the following. First, the eviction rule essentially discretizes the loss into multiples of $\eps$ and we need to perform more refined division to reduce the loss. On the other hand, one can   construct examples showing the pool size grows at least $O(1/\eps)$, since the loss takes range from $[0, 1]$. 
The idea is to bootstrap the baseline algorithm and to reduce the width of experts.\footnote{An expert with loss in $[a, b]$ is said to have width $(b - a)$, see \cite{arora2012multiplicative} for formal definition.}

\paragraph{Precondition the experts} The idea is fairly simple: we just run MWU over the original expert $i$ and the baseline algorithm, and take its output prediction as the new expert $e_{i}$.
Let $\Delta$ denote the loss of the baseline algorithm. Roughly speaking, it is guaranteed that the performance of $e_{i}$ is  at least as good as $i$ and the baseline, and since the baseline has average regret of at most $\eps$, $e_i$ takes loss in $[\Delta-\eps, \Delta]$.\footnote{For technical reasons,  we also need to truncate the loss because it is possible that $e_i$ performs much better than $i$ and the baseline.}
Therefore, the width of the loss is significantly reduced---from $[0, 1]$ to $[0, \eps]$. This allows us to discretize with an additive factor $\eps^2$ instead of $\eps$. 
The idea of preconditioning is ubiquitous and powerful in modern algorithm design, but as far as we know, it is the first time applied to online prediction.

\paragraph{Putting things together} The above bootstrapping procedure reduces the regret from $O(\eps T)$ to $O(\eps^2 T)$, up to some lower order terms.
There is no reason to stop here, and in fact we  repeatedly perform the bootstrapping for roughly $\frac{\log T}{\log n}$ times. Carefully balancing with the lower order terms yields a final regret bound of $O_n\left(T^{\frac{2}{2+\delta}}\right)$. 
At the end, the scheme would maintain a hierarchical pool of experts. The algorithm   plays a  mixed strategy over them, instead of simply applying MWU. Our approach of transforming a weakly sub-linear regret $O(\eps T)$ algorithm to a sub-linear regret $O(T^{1-\alpha})$ algorithm is   general. We believe it  could have broad applications in the area of online prediction.

\subsubsection{Lower bound via learning in games}
Our lower bound draws close connection with learning in games. 
It is well known that one can use no-regret learning algorithm to compute Nash equilibria of a zero-sum game \cite{freund1999adaptive}, via the following template: Alice follows a no-regret algorithm with each of her action as an expert, and Bob (the adversary) best responds to it.
We construct a family of zero-sum games whose equilibrium (or minmax) strategies are far apart, and any single strategy can  only achieve $\eps$-approximate minmax value for a few of them. 
The construction is  by randomly embedding a generalized matching penny game.
Via a counting argument, one can prove that algorithms using sub-linear space simply cannot achieve $\eps$-approximate minmax in all its states for most of the games in the family. This contradicts with the fact that no-regret dynamics achieve minmax value and therefore establishes the lower bound.

\subsection{Compared with Srinivas, Woodruff, Xu \& Zhou}\label{sec:compare}
Before we survey other related work, \cite{stoc22} is the most relevant to us. We discuss our findings in light of the main results therein.  

\paragraph{Lower bound}
On the hardness side, \cite{stoc22} shows that $\Omega(S)$ memory is necessary for achieving $O(\sqrt{nT/S})$ regret even for i.i.d.\ loss sequence, whereas our lower bound indicates that $\Omega(S)$ is necessary for $O(T/S)$ regret against an {\em adaptive} adversary. These two results are incomparable in general, since our lower bound is quantitatively stronger but under a much stronger adversary model. On the technical side, the lower bounds of \cite{stoc22}   leverage communication complexity techniques, whereas ours (\cref{thm:intro-lower}) exploits the connection between no-regret learning and zero-sum games.

\paragraph{Upper bound}
\cite{stoc22} gives an $O(S)$-space algorithm that obtains $\widetilde{O}(\sqrt{nT/S})$ regret assuming that the loss sequence arrives in random order. This matches their aforementioned lower bound. 
Unsurprisingly, the design and analysis of the algorithm hinges heavily upon the random order assumption and does not have any implication for the (standard) worst-case loss. 

For worst-case loss sequence, \cite{stoc22} gives an $O(\frac{n}{\delta T})$-space algorithm with $\delta T$ regret (for any $\delta T = \widetilde \Omega (\sqrt T)$), assuming that the best expert receives a total loss of at most $M =O\left(\frac{\delta^2 T}{\log^2 n}\right)$.
The assumption on the best expert is rather strong: for their result to be meaningful, one needs $\delta < 1$ and so the best expert already has sub-linear loss.
As explained earlier, our \cref{thm:intro-main} does not require any such condition.
In fact, under the assumption, the algorithmic design is conceptually simple. A na\"ive algorithm is to follow the advice of a single expert until its cumulative loss exceeds $M$, then switches to a new one, and repeats. This procedure uses $O(1)$ memory and is worse than the best expert in total loss by at most an  $O(n)$ multiplicative factor. 
Instead of tracking a single expert, \cite{stoc22} designs a more general scheme by sampling a pool of experts and running majority vote. Improving upon the na\"ive idea, their algorithm achieves a   performance of $O(\delta^{-1}\log^2 n)$ multiplicatively worse than the best expert (in terms of total loss). We stress again that our algorithm works beyond this low mistake  regime.


\subsection{Related work}
\paragraph{Identifying best expert is hard} 
We first mention that identifying even an approximately best expert requires $\Omega(n)$ memory 
in stream.  
This is matched by the trivial algorithm of tracking the cumulative loss of all experts. 
The proof is by a simple reduction from the well-studied \textsc{Set-Disjointness} problem;  see the survey \cite{sherstov2014communication}  and references therein. We refer interested reader to the prior work \cite{stoc22} for a detailed argument. 

\paragraph{Prior work on expert learning} 
Forms of the classic MWU algorithm for online learning can be dated back to 1950's \cite{brown1951iterative}. The algorithm has been analyzed in a variety of settings and shown to achieve (nearly) optimal regret \cite{littlestone1989weighted,ordentlich1998cost}. The algorithm also finds a wide range of applications in algorithm design and optimization \cite{cesa2006prediction, freund1997decision,christiano2011electrical,garber2016sublinear,klivans2017learning,hopkins2020robust,ahmadian2022robust}. See the survey \cite{arora2012multiplicative} for a complete treatment.

There are other online learning algorithm that are less computationally expensive than MWU.  
In particular, a line of work  initiated by \cite{kalai2005efficient, hazan2016computational} equips an online learner with an (offline) optimization oracle. The goal is to minimize oracle calls, a proxy for time complexity. 
Strong regret and oracle complexity guarantees can be achieved under   this framework \cite{dudik2020oracle, nika,noah,dai2022follow}.  
Moreover, under various structural  assumptions, one can improve upon the time complexity of MWU \cite{helmbold1997predicting,maass1998efficient,takimoto2001predicting,kalai2005efficient}. These lines of work, however, generally do not consider space bounds. Finally, in terms of technique, the algorithm of \cite{hazan2016computational} also uses the idea of random sub-sampling of the experts.  
 
\paragraph{Memory-efficient online learning} 
There has been a recent spate of work on memory-bounded online learning, mostly in (stochastic) multi-arm bandit settings. 
This includes the study of both regret minimization and pure exploration \cite{liau2018stochastic, chaudhuri2020regret, assadi2020exploration,jin2021optimal,maiti2021multi}. We mention that a recent work \cite{agarwal2022sharp}   considers a multi-pass setting, where the algorithm may take several passes over the data. 
They show that the regret-memory trade-off can be significantly improved when this is allowed. 
Our work is focused on single-pass algorithm.
Space usage is also considered by \cite{on2013} in the analysis of pairwise losses in online learning,  but under a restricted memory model.

A related line of work is on the memory-sample lower bounds for statistical and computational learning problems \cite{steinhardt2016memory,raz2017time,raz2018fast,garg2018extractor, sharan2019memory,gonen2020towards, garg2021memory,marsden2022efficient}, such as learning sparse parities and linear regression. Our problem is not statistical in nature, since we assume a worst-case loss sequence, and therefore does not lie in their setting. 
Other lines of work also study learning problems in data streams, including continual learning \cite{chen2022memory}, entropy estimation \cite{acharya2019estimating,aliakbarpour2022estimation}, detecting   correlations \cite{dagan2018detecting}, robust estimation \cite{diakonikolas2022streaming}, learning simple classifiers \cite{brown2022strong} and matrix rank estimation \cite{crouch2016stochastic}.

\paragraph{Adversary models} Our algorithm assumes an oblivious adversary, which is standard in the literature. 
We remark that our lower bound considers a notion of adaptivity   stronger than the typical ones in the literature, where the adversary cannot access the internal randomness of the algorithm; see, e.g., \cite{DBLP:conf/icml/DekelTA12,DBLP:conf/nips/Cesa-BianchiDS13}. Finally, several recent works show that certain lower bounds under adaptive adversary can be circumvented  by smoothed analysis \cite{NIPS2011_692f93be,noah,tim,nika}.

\subsection{Organization}
The remainder of the paper is organized as follows. \cref{sec:prelim} introduces the notations and necessary technical backgrounds. \cref{sec:baseline} describes and analyzes  our baseline algorithm, which serves as a building block for our full algorithm in \cref{sec:full}. The lower bound is proved in \cref{sec:lower}.
We conclude the paper in \cref{sec:con} by pointing out several future directions.
\section{Preliminaries} 
\label{sec:prelim}

\paragraph{Notations} Throughout the paper, we use $T$ to denote the number of rounds and $n$ to denote the number of experts. We write $[n] := \{1, \ldots, n\}$ and $[n_1, n_2]:= \{n_1, \ldots, n_2\}$. Let $\Delta_n$  denote the collection of  probability distributions over $[n]$.
Unless specified otherwise, all logarithms are base $e$.
 We refer to a word of memory as $O(\log(nT ))$ bits.

\subsection{Online learning with expert advice}
We study the classic online prediction with expert problem under the (standard) oblivious adversary model. 
In this problem, an algorithm makes prediction every round with the advice from experts and with the goal of minimizing its regret with respect to the best expert in hindsight.
Formally, 

\begin{definition}[Online learning with oblivious adversary]
\label{def:oblivious}
An algorithm is initiated with memory $M_1$ and makes prediction for $T$ rounds. 
At the $t$-th iteration ($t \in [T]$), 
\begin{itemize}
    \item The algorithm chooses an expert $i_t \in [n]$ based on its memory $M_t$
    \item The nature reveals the loss vector $\ell_{t} \in [0, 1]^{n}$ ,
    \item The algorithm receives loss $\ell_{t}(i_t) \in [0, 1]$ and updates its memory state $M_{t+1}$.
\end{itemize}

An adversary is said to be oblivious if the sequence of loss vectors $\ell_1, \ldots, \ell_T \in [0,1]^{n}$ are chosen independent of the algorithm's decision. Equivalently, the nature fixes the loss vectors in advance (possibly randomly) and they are unknown to the algorithm. An algorithm is said to use up to $M$ bits of space if $\max_{t\in [T]} |M_{t}| \leq M$.
\end{definition}

\begin{remark}
Strictly speaking, the algorithm can not store the loss vector $\ell_t$ into its memory in the second step, as it already takes $\Omega(n)$ bits. Instead, we allow the algorithm to query the entry of the loss vector $\ell_t$.
\end{remark}

The goal of the algorithm is to minimize the total {\em regret} over $T$ rounds, defined as
\begin{align}\label{eqn:def-regret}
   R(T) := \E\left[\sum_{t\in [T]}\ell_t(i_t)\right] - \min_{i^{*} \in [n]} \sum_{t\in [T]}\ell_t(i^{*}),
\end{align}
where the expectation is taken over the randomness of the algorithm.
We   also consider the average regret, defined as $R(T)/ T$.


We also consider the adaptive adversary model and prove a linear memory lower bound. 

\begin{definition}[Online learning with strong adaptive adversary]
\label{def:adaptive}
An algorithm is initiated with memory $M_1$ and makes prediction for $T$ rounds. 
At the $t$-th iteration (for $t \in [T]$), 
\begin{itemize}
    \item The algorithm commits a distribution $p_t \in \Delta_{n}$ over the experts $[n]$ based on its memory state $M_t$;
    \item The adversary reveals the loss vector $\ell_{t} \in [0, 1]^{n}$ after observing the distribution $p_t$;
    \item The algorithm receives loss $\langle p_t, \ell_t\rangle$ and updates its memory state $M_{t+1}$.
\end{itemize}
\end{definition}
Note we assume the algorithm commits a probability distribution over the experts $[n]$ (instead of a single expert) but the realization is unknown to the adversary. Equivalently, it can be seen that the adversary can choose loss vector $\ell_t$ based on the {\em decisions} as well as the {\em internal randomness} of the algorithm up to round $t-1$.

This notion of adaptivity here is stronger than the traditional one in the online learning literature, where the adversary does not have access to the algorithm's internal states. 
Rather, \cref{def:adaptive} closely resembles white-box adversary for adversarially robust streaming algorithm, proposed recently by \cite{10.1145/3517804.3526228}.
We remark that the standard implementation of MWU takes $\widetilde{O}(n)$ space and achieves $O\left(\sqrt{T\log n }\right)$ regret against strong adaptive adversary.

\subsection{Algorithmic and mathematical tools }

\paragraph{Multiplicative weights update}
\label{sec:MWU}

Our algorithm will use the classic MWU scheme as a subroutine. We state its update and decision rule, and the formal regret guarantees. See \cite{arora2012multiplicative} for a standard exposition.
\vspace*{0.6em}

\begin{algorithm}[H]
\caption{Multiplicative weight update (MWU)}\label{alg:mwu}
    \SetAlgoLined
    \DontPrintSemicolon
    \KwInput{Learning rate $\eta$, expert $[n]$}
    \For{$t$ from $1$ to $T$}{
        Compute $p_t \in \Delta_n$ over experts such that $p_t(i) \propto \exp\left(-\eta \sum_{\tau=1}^{t-1}\ell_\tau(i)\right)$\;
        Sample an expert $i_t \sim p_t$ and observe the loss vector $\ell_t \in [0, 1]^n$
    }
\end{algorithm}

\begin{lemma}[MWU regret, \cite{arora2012multiplicative}]
\label{lem:mwu}
Suppose $n, T, \eta > 0$ and the loss $\ell_t \in [0,1]^n$ ($t\in [T]$), then the multiplicative weight update algorithm satisfies
\begin{align*}
    \sum_{t=1}^{T} \langle p_t, \ell_t\rangle - \min_{i^{*}\in [n]}\ell_t(i^{*}) \leq \frac{\log n}{\eta} + \eta T,
\end{align*}
and with probability at least $1-\delta$,
\begin{align*}
    \sum_{t=1}^{T}  \ell_t(i_t) - \min_{i^{*}\in [n]}\ell_t(i^{*}) \leq \frac{\log n}{\eta} + \eta T + O\left(\sqrt{T \log (n/\delta)}\right).
\end{align*}

Taking $\eta = \sqrt{\frac{\log n}{T}}$, the MWU algorithm has a total regret of $O\left(\sqrt{T\log (nT)}\right)$ with probability at least $1- 1/\poly(T)$ and a standard implementation takes $O(n\log T)$ bits of memory.
\end{lemma}

\paragraph{Probability and concentration inequalities} We state the standard concentration inequality.
\begin{lemma}[Azuma-Hoeffding bound]
Let $X_0, \ldots, X_n$ be a martingale  with respect to the filter $F_0 \subseteq F_1 \cdots \subseteq F_n$ such that for $Y_i = X_i - X_{i-1}$, $i \in [n]$, we have that $|Y_i| = |X_i - X_{i-1}| \leq c_i$. Then
\begin{align*}
    \Pr[|X_t - Y_0 | \geq t ] \leq 2\exp\left(-\frac{t^2}{2\sum_{i=1}^{n}c_i^2}\right).
\end{align*}
\end{lemma}

\paragraph{Minmax theorem}
It is well known that the equilibrium value of a zero sum game is unique, and it equals to the minmax or maxmin value.
\begin{lemma}[Minmax Theorem \cite{neumann1928theorie}]
\label{lem:minmax-thm}
For any $A \in \R^{n\times n}$, the minmax theorem guarantees that
\[
\min_{x \in \Delta_n}\max_{y \in \Delta_{n}}x^{\top}Ay = \max_{x \in \Delta_{n}}\min_{y \in \Delta_{n}} x^{\top}Ay.
\]
\end{lemma}
\section{The building block}
\label{sec:baseline}

We first give an online learning algorithm that achieves $\widetilde O (\eps T + TB^{-1/2}+ \eps^2nB)$ total regret over $T$ days in $S=\widetilde O\left(\eps^{-2}\right)$ space, where $B \ll T$ is a parameter specified later.
This procedure apparently does not achieve our end goal, but instead it will serve as a building block for the our full algorithm in \cref{sec:full}.

\begin{theorem}
\label{thm:weak-sublinear}
Let $T, n, B$ be positive integer, $\eps\in (0,1/2)$, there exists an online learning algorithm that achieves $O\left(\eps T + TB^{-1/2}\log^{1/2} (nT) + \eps^2 n B \log T\right)$ regret with probability at least $1 - 1/\poly(T)$ and uses $O(\eps^{-2}\log^3 (nT))$ bits of memory.
\end{theorem}

For the sake of simplicity, we assume $\eps^{-1}$ and $T/B$ are integers in the rest of section.

\subsection{Baseline algorithm}
\paragraph{Algorithm description} The \textsc{Baseline} (\cref{alg:baseline}) maintains a pool of experts $\mathcal{P}$ at every step. 
It divides the whole sequence into $T/B$ epochs, where each epoch consists of $B$ (contiguous) days. 
\textsc{Baseline} proceeds epoch by epoch. A random set of experts $R_t$ of size $\eps^{-2}$ is sampled uniformly without replacement from $[n]$ and enters the pool at the beginning of each epoch.
For now, assume an expert's cumulative loss within each epoch is tracked and stored, ever since it joins the pool.
Within an epoch, we maintain the same pool of experts and run the MWU algorithm only on them, starting with the uniform weights, and produces a (random) decision every round. 
Na\"ively, the pool size grows by one every epoch, which is unacceptable for large $T$.
To address the issue and bound the pool size, we evict experts at the very end of each epoch by comparing their average losses.

\paragraph{Intuition}
Intuitively, if any expert performs poorly relative to the others in the pool, it makes sense  to evict it. However, care needs to be taken when comparing a long-surviving expert with a recently joined one. 
Due to the worst-case nature of the input, a new expert may start off by receiving significantly less loss  than an old one. Nevertheless, it is  yet unclear that it will continue to excel in the long run. Therefore, we design the algorithm so that an   expert cannot be deleted simply because it is outperformed   by a newer expert. It turns out that this rule is crucial in proving our memory bounds as well. 

\paragraph{Eviction rule} 
More formally, for any epoch $t \in [T/B]$, let $\mathcal P_t \subseteq [n]$ be the pool of experts at the beginning of $t$-th epoch (after adding $R_t$). The pool remains unchanged throughout the $t$-th epoch, and the algorithm considers removing certain experts from the pool   at the end of the epoch.
For experts in $R_t$, the rule is simple and we just keep the best expert. 
For each remaining expert $i$ (including the one that survives in $R_t$), let $\alpha(t,i) \leq t$ be the epoch when expert $i$ enters the pool. 
Let $\Gamma_{t, i} := [\alpha(t,i): t]$ be the period of time from  the $\alpha(t,i)$-th epoch to the $t$-th epoch.
For simplicity we   assume that \textsc{Baseline} explicitly tracks and updates the total loss value of all experts $ i \in \mathcal{P}_t$, over each epoch since their entrance. (We will discuss later how to optimize the space usage.)

Let $\mathcal{L}_{t}(i) = \frac{1}{B}\sum_{b=1}^{B}\ell_{(t-1)B + b}(i)$ be the average loss for expert $i$ in epoch $t$.  For any time interval $\mathcal{I} \subseteq \Gamma_{t,i}$, let $\mathcal{L}_{I}(i) = \frac{1}{|I|} \sum_{t\in I}\mathcal{L}_{t}(i)$ be the average loss of expert $i$ over $I$. 
The algorithm compares $i$ with every other expert $j\in \mathcal{P}_t$. The expert $i$ is evicted at the end of epoch $t$ if and only if 
\begin{enumerate}[(i)] 
    \item The expert $j$ entered the pool $\mathcal{P}_t$ before the expert $i$; {and} 
    \item The average loss $\mathcal{L}_{\Gamma_{t,i} }(i)$ of $i$ over $\Gamma_{t, i}$ is at least that of expert $j$ up to an additive factor of  $\eps$: 
    \begin{equation}\label{eqn:evict-rule}
        \mathcal{L}_{\Gamma_{t,i} }(i) \geq \mathcal{L}_{\Gamma_{t,i} }(j) - \eps.
    \end{equation}
\end{enumerate}
Simply put, condition (i) ensures that an older expert in the pool cannot be kicked out due to a younger one. 
Notice that (ii) effectively requires   (i), since our algorithm only keeps track of the loss of the experts within the pool. 
If $j$ entered the pool later than $i$, we cannot even compute the right-hand side of \eqref{eqn:evict-rule}.

We run the comparison-based pruning procedure at the very end of each epoch and use $\widetilde{P}_t$ to denote the set of experts that survive pruning. 
We will argue that (1) the size of the pool is bounded  and (2) overall the pool contains good experts such that our algorithm is competitive against the best expert, albeit it can easily be outside the pool.

\vspace*{0.5em}

\begin{algorithm}[H]
\caption{Baseline expert learning algorithm (\textsc{Baseline})}\label{alg:baseline}
    \SetAlgoLined
    \DontPrintSemicolon
    \KwInput{Parameter $T, B$ and $\eps$, experts $[n]$}
    \For{each epoch $t = 1,2, \ldots, T/B$}{
    Sample a random set $R$ of $\eps^{-2}$ experts without replacement from $[n]$ and add them to $\mathcal{P}$\;
    Initialize MWU (with uniform weights) over experts in $\mathcal{P}$\;
    \For{each day $b = 1,2, \ldots, B$}{
    Sample and play the MWU decision over experts in $\mathcal{P}$
    }
    \tcc{Evict expert}
    Remove all except the best expert from $R$\;
    \For{every pair $\{i,j\} \in \mathcal{P}$}{
            Remove $i$ from $\mathcal{P}$ if $j$ entered $\mathcal{P}$ before $i$ and condition \eqref{eqn:evict-rule} holds
    }
    }
\end{algorithm}

\paragraph{Implementation details} 
Na\"ively, the algorithm stores the cumulative loss of each expert in the pool, over every epoch since it entered.
This may cause large memory usage.
However, observe that, to execute the eviction rule, it is only required that for each $j$, we have the data $\mathcal{L}_{\Gamma_{i,t} }(j)$ for all $i$ that entered later than $j$, in addition to its own cumulative loss. 
Therefore, the algorithm can  explicitly track these values only; and if the pool size is $S$, then this takes $O(S^2)$ words of memory. See \cref{fig:saving} for an illustration.

\begin{figure}[ht]
    \centering
    \includegraphics[scale=0.6]{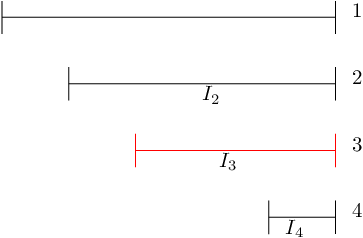}
    \caption{Pool at    the end of an epoch. If expert $3$ is removed, then we no longer store $\mathcal{L}_{I_3}(1),  \mathcal{L}_{I_3}(2)$.}
    \label{fig:saving}
\end{figure}

\subsection{Analysis of the Baseline Algorithm}
\label{sec:an}

We now provide a formal analysis of \textsc{Baseline} (\cref{alg:baseline}) and argue its memory and regret guarantees. 
\subsubsection{Memory bound}
\label{sec:memory-baseline}
The algorithm does not dictate an explicit bound on the memory, and in particular, on the size of pool.
First, we present a key technical lemma that insists a \textit{loss vs.\ length} structure on the surviving experts $\widetilde{\mathcal{P}}_t$. 
For a fixed $t$ and any $i,j \in \widetilde{\mathcal{P}}_t$,  write $i \prec j$ if $\alpha(t, i) > \alpha(t, j)$---namely, $i$ joined the pool later than $j$, so $j$ is older---and $i \succ j$ otherwise. Note that expert $i$ and $j$ must join the pool at different time due to our eviction rule. 
Roughly speaking, the lemma states that if $i$ survives the pruning (after Line 10 of \cref{alg:baseline}), then for any $j \succ i$,  one of the following must happen: 
\begin{enumerate}[(i)]
    \item expert $j$ suffers significantly more average loss over $\Gamma_{t,j}$ than $i$ suffers overs  its interval $\Gamma_{t, i}$; or 
    \item expert $j$ has resided in the pool significantly longer than $i$.
\end{enumerate} 
For notational convenience, for any  $ i,j \in \widetilde{\mathcal{P}}_t$ such that $j\succ i$, define $L_{i,i} =\mathcal{L}_{\Gamma_{t,i} }(i) $, $L_{i,j} =\mathcal{L}_{\Gamma_{t, i} }(j) $ and accordingly $L_{j,j} =\mathcal{L}_{\Gamma_{t, j} }(j) $. Then we have
\begin{lemma}[loss vs.\ length]\label{lem:loss-length}
    For any epoch $t \in [T/B]$, suppose  experts $i, j\in \widetilde{\mathcal{P}}_t$ and  $j\succ i$.  Let $\alpha \in (0, 1)$, then at least one of the following must hold:
    \begin{enumerate}[(i)]
        \item $L_{j,j} \geq  L_{i,i} + \epsilon - \alpha$;
        \item $|\Gamma_{t,j}|  \geq \left(1 + \frac{\alpha}{1 - \alpha}\right)|\Gamma_{t,i}|$.
    \end{enumerate}
\end{lemma}
    
The underlying intuition is simple: Since $i$ survives the pruning and $j$ joins the pool earlier than $i$, we know that $j$ achieves $\eps$ worse average performance than $i$ over $i$'s own interval. Now consider the opposite of condition (i)---it asks $j$ to be at most $\eps - \alpha$ worse average loss overall than $i$. For that to happen, $j$ needs to have low error on the days before $i$ enters the pool. However, even if $j$ gets $0$ loss over this prior period, it still requires time to bring   the average loss down by at least $\eps$. Therefore, $j$ must have entered the pool quite earlier than $i$, which is condition (ii).

\begin{proof}[Proof of \cref{lem:loss-length}]
    By definition of the algorithm, any such expert $i$ must have survived the comparison-based pruning against $j \succ i$. Therefore, the condition \eqref{eqn:evict-rule} must fail for $i,j$, and so we have 
    \begin{equation}\label{eqn:fail-rule}
         L_{i,j} > L_{i,i} + \eps.
    \end{equation}
    Now fix $\alpha \in (0, 1)$ and let's assume that $L_{j,j} < L_{i,i} +\eps - \alpha$. It suffices to show that $|\Gamma_{t, j}|  \geq (1 + \frac{\alpha}{1 -\alpha})|\Gamma_{t, i}|$. Let $T_1 = |\Gamma_{t, j}| -|\Gamma_{t, i}| > 0$  be the number of extra days  $j$ lies in the pool than $i$, and $T_2 = |\Gamma_{t,i}|$. Also let $L_1 = \mathcal{L}_{\Gamma_{t,j} \setminus \Gamma_{t,i}}(j)$  be    the average loss of $j$ over the $T_1$ days. Simply rewriting the assumption that $L_{j,j} < L_{i,i} +\eps - \alpha$:
    \begin{equation}\label{eqn:assumption1}
        L_{j,j} = \frac{L_1 T_1 + L_{i,j} T_2}{T_1 + T_2} < L_{i,i} + \eps - \alpha.
    \end{equation}
    Since $L_1 \geq 0$, we get 
    \begin{equation}
        \frac{L_{i,j} T_2}{T_1 + T_2} < L_{i,i} + \eps -\alpha.
    \end{equation}
    Substituting the inequality $L_{i,i} < L_{i,j} -\eps$ \eqref{eqn:fail-rule} to the right-side:
    \begin{equation}
        \frac{L_{i,j} T_2}{T_1 + T_2} < L_{i,j} - \alpha
    \end{equation}
    Rearranging, we get 
    \begin{equation}
        T_1 >  \left(\frac{L_{i,j} }{L_{i,j}-\alpha} -1\right)T_2 =   \frac{ \alpha  }{L_{i,j}-\alpha} T_2 > \frac{\alpha}{1 - \alpha} T_2,
    \end{equation}
    where the last inequality follows since $L_{i,j} \leq 1$. Equivalently, we have that $ |\Gamma_{t,j}| -|\Gamma_{t,i}| > \frac{\alpha}{1 - \alpha}|\Gamma_{t,i}|  $, and   this completes the proof.
\end{proof}

Using the above lemma, we can bound the size of the pool via a potential function argument. The potential takes both situations of \cref{lem:loss-length} into account, where either length or loss gets larger. 

\begin{lemma}[pool size]
\label{lem:pool-size}
   For any epoch $t \in [T/B]$, the size of the pool $\widetilde{P}_t$ is at most $S =  \frac{4}{\eps}\log T$.
\end{lemma}
\begin{proof}
Fix $t$ and let $S$ be the size of the pool.  We sort the experts in the pool in ascending order  of their entering times: $i_1 \prec i_2  \prec \cdots\prec i_S$. Define the potential function $\Phi : [S] \rightarrow \R_{\geq 0}$, where 
\begin{align}
\label{eq:potential}
\Phi(\tau) = 2 \log |\Gamma_{t, i_{\tau}}| +  L_{i_{\tau}, i_{\tau}}, \quad \tau \in [S].
\end{align}    

We note that $\Phi(1) \geq 0$ and $\Phi(S) \leq 2\log T + 1$. The goal is to prove 
\begin{align}
\label{eq:potential-increase}
    \Phi(\tau + 1) - \Phi(\tau) \geq \eps, \quad \forall \tau \in [S-1].
\end{align}
This would imply the pool size $S \leq 4\eps^{-1}\log T$. 

We observe that \cref{eq:potential-increase} is simply true whenever $L_{i_{\tau+1}, i_{\tau+1}} \geq L_{i_{\tau}, i_{\tau}} + \eps$ since $i_{\tau+1}$ enters the pool before $i_{\tau}$ by our assumption.
Now, it is safe to write $L_{i_{\tau+1}, i_{\tau+1}} = L_{i_\tau, i_{\tau}} + \eps - \alpha$, for some $\alpha \in (0,1)$ ($\alpha$ would not exceed $1$ as one can easily show $L_{i_\tau, i_{\tau}} + \eps < 1$, i.e., $i_{\tau}$ survives at $\widetilde{P}_t$). Then we have
\begin{align*}
    \Phi(\tau + 1) - \Phi(\tau) &=  2\log(|\Gamma_{t, i_{\tau+1}}|/|\Gamma_{t, i_{\tau}}|) + L_{i_{\tau+1}, i_{\tau+1}} - L_{i_{\tau}, i_{\tau}} \\
    &= 2\log \left(1 + \frac{\alpha}{1 - \alpha}\right) + \eps - \alpha\\
    &\geq  \min \left\{\frac{\alpha}{1 - \alpha}, 2\log 2 \right\} + \eps - \alpha \geq \eps,
\end{align*}
where the second step follows from   \cref{lem:loss-length}, the third step follows from $\log(1+x)\geq \frac{x}{2}$ whenever $x<1$ and the last step holds since $\alpha \in (0,1)$. We have proven \cref{eq:potential-increase} and completed the proof.    \end{proof} 

We can now wrap up with a memory bound.
\begin{proposition}[memory bound]
    \label{prop:memory}
    At any time during the execution of \textsc{Baseline}, the memory usage is at most  $O\left(\eps^{-2} \log^3 (nT)\right)$ bits.
\end{proposition}
\begin{proof}
Fix an epoch $t$, as we observed earlier, for each expert $j \in \widetilde{\mathcal{P}}_{t-1}$, the algorithm only needs to keep track of their loss over the intervals $\Gamma_{t-1, i}$ for all  $i \in \mathcal{P}_{t-1}$ such that $i \prec j $. In particular, for each $j 
  \in \mathcal{P}_{t-1}$, the algorithm stores $\mathcal{L}_{\Gamma_{t-1,i}} (j)$ for all $i \prec j$. 
  This is sufficient for executing the eviction rule (\cref{eqn:evict-rule}). 
  Storing each $\mathcal{L}_{\Gamma_{t-1, i}} (j)$ takes $O(\log nT)$ bits.
  By \cref{lem:pool-size}, there are at most $S = \frac{4}{\eps} \log T$ experts in $\mathcal{P}_t$. 
  This leads to $S^2 \cdot O(\log nT)= O(\eps^{-2} \log^3 (nT))$ bits of memory usage.
  We sample $|R_t| = \eps^{-2}$ new experts at the beginning of $t$-th epoch and keep track of them within the epoch, these takes $O(\eps^{-2}\log nT)$ extra bits.
\end{proof}

\subsubsection{Regret bound}
\label{sec:baseline-regret}
We now prove a regret bound of {\sc Baseline}. On a high level, the MWU procedure only guarantees that \textsc{Baseline} is always competitive with the best expert {\em in the pool}. The key challenge here therefore is to argue  that our pool is competitive against even the best expert among $[n]$, which may lie outside the pool.
Indeed, next we will prove that the   the experts in the pool is nearly competitive against the best expert, except  in $\widetilde{O}(\eps^2 n)$ epochs. Formally, we show:

\begin{proposition}[regret bound]
    \label{prop:regret}
    Given the parameter $\eps \in (0,1/2)$ and positive integer $B \ll T$, \textsc{Baseline} (\cref{alg:baseline}) achieves a total regret of $O\left(\eps T + TB^{-1/2}\log^{1/2} (nT) + \eps^2 n B \log T\right)$ with probability at least $1 - 1/\poly(T)$.
\end{proposition}

Let $i^*$ denote the best expert in hindsight. Since the adversary is oblivious to the algorithm's decision, it suffices to fix a loss sequence $\ell_1,\ell_2, \ldots, \ell_T \in [0, 1]^n$ and prove the algorithm achieves low regret on it. 
Let $\xi_{t}$ denote the random bits used by {\sc Baseline} during the $t$-th epoch, which includes both the randomness of sampling $R_t$ and the random bits used by MWU within the epoch.

Initialize $\mathcal{B} = \emptyset $ and $i(t) = \mathsf{nil}$ for each $ t \in [T/B]$. We will build up  the set $\mathcal{B} \subseteq [T/B]$ over time. Intuitively, it   contains  ``unlucky'' epochs that we have no control over the regret. On the other hand, for any lucky epoch $t\in [T/B]\backslash \mathcal{B}$, we would associate the $t$-th epoch with an expert $i(t) \in [n]$ that   (1) is competitive with $i^{*}$, and (2) lies in the pool $\mathcal{P}_{t}$.
Formally, the value of $\mathcal{B}$ and $\{i(t)\}_{t\in [T/B]}$ are   assigned by the following stochastic process.

\paragraph{Stochastic process} Starting with $\beta(1) = 1$ and $\tau = 1$, define a (discrete) stochastic process by   realizing the randomness $\xi_1, \ldots, \xi_{T/B}$ epoch by epoch.
Suppose the process proceeds to step $\tau$ and the randomness $\xi_{1}, \ldots, \xi_{\beta(\tau) - 1}$ are realized up to the $(\beta(\tau) - 1)$-th epoch. 
Then the pool $\widetilde{\mathcal{P}}_{\beta(\tau) - 1}$ is also fixed by definition.
Let $t(i^{*}, \tau)$ denote the epoch when $i^{*}$ gets evicted, conditioned on $i^{*}\in R_{\beta(\tau)}$ due to the randomness $\xi_{\beta(\tau)}$ and that it survives the competition among $R_{\beta(\tau)}$ (Line 7 of  \cref{alg:baseline}). 
(If $i^{*}$ never gets evicted, then we set $t(i^{*}, \tau) = \infty$.)
The (conditional) eviction time $t(i^{*}, \tau)$ of expert $i^{*}$ is determined by  the fixed  pool $\widetilde{\mathcal{P}}_{\beta(\tau) - 1}$ and the loss sequence $\ell_1, \ldots, \ell_{T}$. 
In other words, we observe that $t(i^{*}, \tau) \in [T/B]$ is only a function of  $\ell_1, \ldots, \ell_{T}$ and $\widetilde{\mathcal{P}}_{\beta(\tau) - 1}$: it is independent of $\xi_{\beta(\tau) +1}, \ldots, \xi_{T/B}$ because $i^{*}$ enters at epoch $\beta(\tau)$ and can only be kicked out by experts joining before it, i.e., the experts in $\widetilde{\mathcal{P}}_{\beta(\tau) - 1}$; and it is independent of $\xi_{\beta(\tau)}$ because we already condition on the event of $i^{*}$ surviving among $R_{\beta(\tau)}$.

We now continue to define the stochastic process  and consider the following cases.
\begin{itemize}
    \item Case 1. Suppose $t(i^{*}, \tau) \neq \infty$, i.e., expert $i^{*}$ would be evicted at the end of epoch $t(i^{*}, \tau) \in [\beta(\tau), T/B]$. Suppose it is removed  by expert $i_{\tau}^*$. Then we set $\beta(\tau+1) = t(i^{*}, \tau)+1$ and assign $i(\beta(\tau)) = \cdots = i(\beta(\tau+1)-1) = i_{\tau}^*$. 
    We then realize the randomness $\xi_{\beta(\tau)}, \ldots, \xi_{\beta(\tau+1)-1}$ and proceed to step $\tau + 1$.
    \item Case 2. Suppose $t(i^{*}, \tau) = \infty$, i.e., expert $i^{*}$ would not be kicked out of the pool once it is sampled in $R_{\beta(\tau)}$ and survives the competition among $R_{\beta(\tau)}$. We then realize the randomness of $\xi_{\beta(\tau)}$ and further divide into two cases based on it.
    \begin{itemize}
        \item Case 2-1. If $R_{\beta(\tau)}$ contains an expert $i_{\tau}^* \in [n]$ that receives less loss than $i^{*}$ during epoch $\beta(\tau)$, then we set $\beta(\tau + 1) = \beta(\tau) + 1$ and assign $i(\beta(\tau)) = i_{\tau}^*$. We proceed to step $\tau+1$.
        \item Case 2-2. Suppose $R_{\beta(\tau)}$ does not contain any expert that gets less loss than $i^{*}$ during epoch $\beta(\tau)$, then 
        \begin{itemize}
            \item Case 2-2-1. If expert $i^{*}$ has been sampled, i.e., $i^{*} \in R_{\beta(\tau)}$, then assign $i(\beta(\tau)) = \cdots =i(T/B)= i^*$ and terminate the process. We note that by the definition of $t(i^{*}, \tau)$ and the condition, $i^{*}$ will survive till the end.
            \item Case 2-2-2.  If expert $i^{*}$ has not been sampled, i.e., $i^{*}\notin R_{\beta(\tau)}$, then add $\beta(\tau)$ into $\mathcal{B}$ and set $\beta(\tau + 1) = \beta(\tau) + 1$. We proceed to step $\tau + 1$.
        \end{itemize}
    \end{itemize}
\end{itemize}

Having defined the stochastic process, we proceed with our regret analysis.
The following two lemmas are critical to the proof. 
First, we show the size of $\mathcal{B}$ is small with high probability:
\begin{lemma}[unlucky epoch]
\label{lem:unluck}
    With probability at least $1 - 1/\poly(T)$, the stochastic process ends with $|\mathcal{B}| \leq O(\eps^{2}n\log T)$.
\end{lemma}
\begin{proof}
    We count the total number of steps that the stochastic process come with Case 2-2. 
    Whenever the process falls into Case 2-2, i.e., at some step $\tau$, $R_{\beta(\tau)}$ does not contain any expert better than $i^{*}$ during epoch $\beta(\tau)$, we know that 
    \begin{align}
    \Pr\left[i^{*} \in R_{\beta(\tau)} | R_{\beta(\tau)} \text{ has no expert better than } i^{*} \text{ during epoch } \beta(\tau)\right] \geq \frac{1}{n}\cdot \eps^{-2}, \label{eq:terminate-prob}
    \end{align}
    since there are $\eps^{-2}$ experts   sampled without replacement from $[n]$.
    Consider the following two cases:
    (1) if $i^{*} \in R_{\beta(\tau)}$, then the process would terminate immediately and there will be no more Case 2-2 in the future. This situation happens with probability at least $\frac{1}{n}\cdot \eps^{-2}$ by~\cref{eq:terminate-prob};
and     (2) if $i^{*} \notin R_{\beta(\tau)}$, then the process would still continue.
    Since we do not augment $\mathcal{B}$ in Case 1 and Case 2-1, the size of $\mathcal{B}$ is bounded by total number of steps of Case 2-2, and we have,
    \[
    \Pr\left[|\mathcal{B}| \geq c \eps^2 n \log T\right] \leq \left(1 - 1/\eps^2n\right)^{c\eps^{2}n\log T} \leq T^{-c}.
    \]
    This concludes the proof.
\end{proof}

We then prove for epoch $t \in [T]$, the expert $i(t)$ resides in pool $\mathcal{P}_t$ and is competitive with $i^{*}$ (over certain time period).
\begin{lemma}[cover]
\label{lem:cover}
For any epoch $t \in [T/B] \backslash \mathcal{B}$, suppose $\beta(\tau) \leq t < \beta(\tau + 1)$, then we have
\begin{enumerate}[(i)]
    \item $i(\beta(\tau)) = \cdots = i(t) = \cdots = i(\beta(\tau+1) -1) = i_{\tau}^*  \neq \mathsf{nil}$;
    \item $i_{\tau}^* \in \mathcal{P}_\nu$ for any $\nu \in [\beta(\tau): \beta(\tau+1) -1]$;
    \item $\sum_{\nu=\beta(\tau)}^{\beta(\tau +1) -1}\mathcal{L}_{\nu}(i(t)) < \sum_{\nu=\beta(\tau)}^{\beta(\tau +1) -1}\mathcal{L}_{\nu}(i^*) + \eps (\beta(\tau+1) - \beta(\tau))$.
\end{enumerate} 
\end{lemma}
\begin{proof}
    The first claim is straightforward from the assignment process. Since $t \in [\beta(\tau): \beta(\tau+1) -1]$ and $t \notin \mathcal{B}$, we note in the $\beta(\tau)$-th epoch, the process falls into Case 1, Case 2-1 or Case 2-2-1. 
    
    In Case 1. If $i^{*}$ is sampled in the $R_{\beta(\tau)}$ and happens to survive the comparison among $R_t$, then $i^{*}$ wound be evicted at the end of $(\beta(\tau+1) - 1)$-th epoch when comparing with expert $i_\tau^{*}$. First of all, this indicates that $i_{\tau}^*$ enters the pool before epoch $\beta(\tau)$ (an expert can only be evicted by older expert), and therefore, the eviction time of $i_{\tau}^{*}$ is already determined given $\xi_1, \ldots, \xi_{\beta(\tau)-1}$ and $\ell_1, \ldots, \ell_{T/B}$ and it is no earlier than $\beta(\tau+1) - 1$ because otherwise, it could not kick $i^{*}$ out. The third claims follows directly from the eviction rule (see \cref{eqn:evict-rule}).

    In Case 2-1, we know that $\beta(\tau) = t$, $\beta(\tau+1) = t+1$ and $i_{\tau}^{*} \in R_t$ is an expert that has better performance than $i^*$. The second and last claims are then straightforward.
    
    Finally, in Case 2-2-1, we note that $i_{\tau}^{*} = i^*$ and $i^{*}$ survives till the end. The last two claims are straightforward and we finish the proof.
\end{proof}

Now we can finish the proof of \cref{prop:regret}.
\begin{proof}[Proof of   \cref{prop:regret}]
 Conditioned on the high probability event of  \cref{lem:unluck}, the total regret of \textsc{Baseline} is controlled as follows:
\begin{align*}
    &~ \sum_{t=1}^{T/B}\sum_{b=1}^{B}\ell_{(t-1)B + b}\left(i_{(t-1)B + b}\right) - \ell_{(t-1)B + b}(i^{*})\\
    = &~\sum_{t\in[T/B]\backslash \mathcal{B}}\sum_{b=1}^{B}\left(\ell_{(t-1)B + b}\left(i_{(t-1)B + b}\right) - \ell_{(t-1)B + b}(i^{*})\right) + \sum_{t\in\mathcal{B}}\sum_{b=1}^{B}\ell_{(t-1)B + b}\left(i_{(t-1)B + b}\right) - \ell_{(t-1)B + b}(i^{*}) \\
    \leq &~ \sum_{t\in[T/B]\backslash \mathcal{B}}\sum_{b=1}^{B}(\ell_{(t-1)B + b}(i_{(t-1)B + b}) - \ell_{(t-1)B + b}(i^{*})) + O(\eps^2 n \log T)\cdot B \\
    \leq &~ \sum_{t\in[T/B]\backslash \mathcal{B}}B(\mathcal{L}_{t}(i(t)) - \mathcal{L}_{t}(i^{*})) + \frac{T}{B} \cdot O\left(\sqrt{B\log (nT)}\right) + O\left(\eps^2 n B \log T\right)\\
    = &~ \sum_{\tau}\sum_{t \in [\beta(\tau): \beta(\tau + 1) -1], t\notin \mathcal{B}}B(\mathcal{L}_{t}(i(t)) - \mathcal{L}_{t}(i^{*})) + O\left(TB^{-1/2}\log^{1/2} (nT)\right) + O\left(\eps^2 n B \log T\right)\\
    \leq &~ \sum_{\tau} \eps B (\beta(\tau+1) - \beta(\tau)) + O\left(TB^{-1/2}\log^{1/2} (nT)\right) + O\left(\eps^2 n B \log T\right)\\
    = &~ O\left(\eps T + TB^{-1/2}\log^{1/2} (nT) + \eps^2 n B \log T\right).
\end{align*}
We split the regret based on whether $t$ belongs to $\mathcal{B}$ in the first step. The second step follows from $|\mathcal{B}|\leq O(\eps^2 n \log T)$ (\cref{lem:unluck}) and $\ell_{t} \in [0,1]^{n}$.  
The third step follows from the guarantee of MWU (\cref{lem:mwu}) and the fact that the expert $i(t)$ is contained in the pool $\mathcal{P}_t$ (second claim of   \cref{lem:cover}).
We split $[T/B]$ according to $\beta(\tau)$ in the fourth step and make use of the first claim of   \cref{lem:cover}. The fifth step follows from the third claim of   \cref{lem:cover}.
We conclude the regret analysis.
\end{proof}

Combining \cref{prop:memory} and \cref{prop:regret}, we conclude the proof of \cref{thm:weak-sublinear}.
Moreover, balancing the last two regret terms by taking $B = \left(T/\eps^2n\right)^{2/3}$, we   get:
\begin{corollary}
\label{cor:weak-sublinear}
Let $T, n$ be positive integer, $\eps\in (0,1/2)$, $T = \Omega(\eps^2 n)$, there exists an online learning algorithm that achieves $\widetilde{O}(\eps T + T^{2/3}(\eps^2n)^{1/3})$ regret with probability at least $1 - 1/\poly(T)$ and uses $O(\eps^{-2}\log^3 (nT))$ bits of memory.
\end{corollary}

\section{Full algorithm and analysis} \label{sec:full}

Building upon \textsc{Baseline}, we can state our main result.

\begin{theorem}[Main algorithmic result]
\label{thm:sublinear}
Let $T, n$ be positive integers and $\delta \in (0, 1]$. There exists an online learning algorithm that achieves a total regret of $\widetilde{O}\left(n^2T^{\frac{2}{2 +\delta}}\right)$ with probability at least $1 - \poly(T)$ and uses $O\left(n^{\delta} \log^4(nT)\right)$ bits of memory.
\end{theorem}

\subsection{Full Algorithm}\label{sec:full-algo}

\paragraph{Parameters} Let $\eps = n^{-\delta/2}$ and $T = n(n/\eps)^{K}$. For simplicity, we assume $\eps^{-1}$, $\eps n$ and $K$ are positive integers for now. 
Define $T_k = n(n/\eps)^{k}$ for each $k \in [K]$, and the epoch length is fixed as $B = 1/\eps^2$.
\paragraph{Algorithm description}
The $\textsc{FullAlgo}$ (pseudocode in \cref{alg:full}) aggregates $\textsc{Baseline}_{+}$ (\cref{alg:baseline+})   by levels. 
We first describe the algorithm $\textsc{Baseline}_{+}$, which takes in a level parameter $k$.
 At the bottom level ($k = 1$), the algorithm repeatedly runs {\sc Baseline} for $T/T_1 = (n/\eps)^{K-1}$ episodes. 
Within each episode, the algorithm starts with a fresh run of {\sc Baseline} and continues for $T_1 = n^2/\eps$ days. The $T_1$ days are split into $\eps n^2$ epochs, and each epoch contains $B = 1/\eps^2$ days.
We will later see that this guarantees that each $n^2/\eps$ days, the algorithm gets a total regret of $O(n^2\log (nT))$ compared with the best expert.

$\textsc{Baseline}_{+}$ differs significantly from $\textsc{Baseline}$ starting from the second level ($k \geq 2$). 
Instead of potentially playing a different decision every day, 
$\textsc{Baseline}_{+}$ joins every $T_{k-1} = n(n/\eps)^{k-1}$ consecutive days as one {\em decision day} and plays the same decision on all of them. There are $T/T_{k}$ episodes and $n/\eps$ decision days within each episode. Again, they are divided into $\eps n$ epochs with $1/\eps^{2}$ days each epoch. The algorithm restarts every episode.

The key point is that instead of directly following the advice of expert $i$, the algorithm follows from the combination of expert $i$ and $\textsc{Baseline}_{+}(k-1)$.
In particular, $\textsc{MergeExp}$ (pseudocode in \cref{alg:merge}) takes an expert $i$ and $\textsc{Baseline}_{+}(k-1)$ and runs MWU over them. 
We take it as the new expert $e_{k, i}$ for level $k$. The advantage is that the loss of $e_{k, i}$ is roughly the minimum of $\textsc{Baseline}_{+}(k-1)$ and $i$ (by the regret guarantee of MWU) and thus has small width.
This motivates the modified eviction rule.

\paragraph{Eviction rule}
We reload the notations from  \cref{sec:baseline} and introduce a few more.
For any level $k \in [K]$, episode $r \in [T/T_{k}]$, epoch $t \in [\eps n]$, let $P_{k, r, t} \subseteq \{e_{k, i}\}_{i \in [n]}$ be the pool of experts at the beginning of $t$-th epoch (after adding $R_{k, r, t}$) and the pool remains unchanged during the $t$-th epoch.
Let $T_{k,r,t,b} = (r-1)T_{k} +(t-1)BT_{k-1} + (b-1)T_{k-1}$.
For each decision  day $b \in [B]$, we update the MWU and the cumulative loss according to the {\em truncated loss}
\begin{align}
\widehat{\mathcal{L}}_{k,r,t, b}(e_{k, i}) = \max\left\{\mathcal{L}_{k,r,t,b}(e_{k, i}) - \mathcal{L}_{k,r,t,b}\left(\textsc{Baseline}_{+}(k-1)\right),  - \eps^{k-1}\log^{2k-1}(nT) \right\}
\end{align}
where
\begin{align}
    \mathcal{L}_{k,r,t, b}(e_{k, i}) = \frac{1}{T_{k-1}}\sum_{\tau=1}^{T_{k-1}}\ell_{T_{k,r,t,b} + \tau}(e_{k, i})
\end{align} 
is the average loss on the $b$-th decision  day,
Note that without the truncation, $\widehat{\mathcal{L}}_{k,r,t,b}$ would simply be a shift of $\mathcal{L}_{k,r,t, b} \in [0, 1]^{n}$. Looking ahead, the truncation guarantees the width of $\widehat{\mathcal{L}}_{k,r,t, b}$ to be $2\eps^{k-1}\log^{2k-1}(nT)$, since it is possible that $e_{k, i}$ performs much better than expert $i$ and $\textsc{Baseline}_{+}(k-1)$.

Now we can state the eviction rule as follows. 
For experts in $R_{k,r,t}$, we just keep the best expert. 
For each remaining expert $e_{k, i}$ (include the one that survives in $R_{k,r,t}$), let $\Gamma_{k,r, t, i}$ be the period of time that $e_{k, i}$ resides in the pool.
For any time interval $I \subseteq \Gamma_{k,r,t, i}$, recall that $\mathcal{L}_I(e_{k, i}) = \frac{1}{|I|} \sum_{t\in I}\mathcal{L}_{t}(e_{k, i})$ is the average loss of expert $j$ over $I$.  
Let $\widehat{\mathcal{L}}_I(e_{k, i}) = \frac{1}{|I|} \sum_{t\in I}\widehat{\mathcal{L}}_{t}(e_{k, i})$ be the cumulative truncated loss defined similarly.
The algorithm compares $e_{k, i}$ with every other expert $e_{k, j}\in \mathcal{P}_{k,r,t}$, and the expert $e_{k, i}$ is evicted at the end of epoch $t$ if and only if 
\begin{enumerate}[(i)] 
    \item The expert $e_{k,j}$ entered the pool $\mathcal{P}_{k,r,t}$ before the expert $e_{k, i}$; {and} 
    \item The average loss $\widehat{\mathcal{L}}_{\Gamma_{k, r,t,i}}(e_{k, i})$ of $e_{k, i}$ over $\Gamma_{k,r,t,i}$ is at least that of expert $e_{k, j}$ up to an additive factor of  $\eps^{k}\log^{2k-1}(nT)$: 
    \begin{equation}\label{eqn:new-evict-rule}
        \widehat{\mathcal{L}}_{\Gamma_{k,r, t, i} }(e_{k, i}) \geq \widehat{\mathcal{L}}_{\Gamma_{k,r,t,i} }(e_{k, j}) - \eps^{k}\log^{2k-1}(nT).
    \end{equation}
\end{enumerate}

In summary, $\textsc{Baseline}_{+}(k)$ (for $k \in [K]$) differs from $\textsc{Baseline}$ in three ways:
\begin{itemize}
    \item $\textsc{Baseline}_{+}(k)$ restarts every $T_k$ days, and within each episode, it regards $T_{k-1}$ days as one decision day;
    \item $\textsc{Baseline}_{+}(k)$ follows   the decision of $\textsc{MergeExp}$ instead of directly using the original experts $[n]$, and crucially it considers the truncated loss for eviction and MWU update;
    \item The eviction threshold changes from $\eps$ to $\eps^{k}\log^{2k-1}(nT)$.
\end{itemize}

Finally, we note that \textsc{FullAlgo} outputs the decision of $\textsc{Baseline}_{+}(K)$.
The pseudocode of these procedures are given below.

\begin{algorithm}[H]
\caption{$\textsc{Baseline}_{+}$}\label{alg:baseline+}
    \SetAlgoLined
    \DontPrintSemicolon
    \KwInput{Parameter $k$}
    \For{each episode $r = 1,2, \ldots, T/T_{k}$}{
    \tcc{if $k=1$, then loop from $1$ to $\eps n^2$}
    \For{each epoch $t = 1,2,\ldots, \eps n$}{
    For all $i\in [n]$, let $e_{i,k} = \textsc{MergeExp}(k,i)$.\;
    Sample a random set $R$  of $\eps^{-2}$ experts without replacement from $\{e_{i, k}\}_{i \in [n]}$, add them to $\mathcal{P}$.\;
    \For{each decision day $b = 1,2, \ldots, B$}{ 
    Compute $p \propto \exp\left(-\eta \sum_{\tau=1}^{b-1} \widehat{\mathcal{L}}_{k,r,t, \tau}(e_{k, i})\right)$ for $e_{k, i} \in \mathcal{P}$\;
    Sample an expert $i_{k,r,t, b} \sim p$ and follow $e_{k, i_{k,r,t,b}}$ for $T_{k-1}$ days\;
    }
    Remove all except the best expert from $R$\;
    \For{every pair $\{e_{k, i},e_{k, j}\} \in \mathcal{P}$}{
            Remove $e_{k, i}$ from $\mathcal{P}$ if $e_{k, j}$ entered $\mathcal{P}$ before $e_{k, i}$ and condition \eqref{eqn:new-evict-rule} holds
    }
    }
    Clear the pool $\mathcal{P}$ and restart
    }
\end{algorithm}

\begin{algorithm}[H]
\caption{Merge expert ($\textsc{MergeExp}$)}\label{alg:merge}
    \SetAlgoLined
    \DontPrintSemicolon
    \KwInput{Parameter $k$, expert $i$}
    \KwOutput{Expert $e_{k, i}$}
    \For{each episode $r =1,2, \ldots, T/T_{k-1}$}{
    Initiate with uniform weight over expert $i$ and $\textsc{Baseline}_{+}(k-1)$\;
    \For{$t = 1,2,\ldots, T_{k-1}$}{
    Run MWU over expert $i$ and $\textsc{Baseline}_{+}(k-1)$, and play the decision
    }
    }
\end{algorithm}

\begin{algorithm}[H]
\caption{Full expert learning algorithm ($\textsc{FullAlgo}$)} 
\label{alg:full}
    \SetAlgoLined
    \LinesNumberedHidden
    \DontPrintSemicolon
    \KwInput{Parameter $T$, $\eps$, experts $[n]$}
    Maintain $\textsc{Baseline}_{+}(k)$ (for each $k \in [K]$) and play the decision of $\textsc{Baseline}_{+}(K)$
\end{algorithm}

\subsection{Analysis of \textsc{FullAlgo}}

We provide a formal analysis of $\textsc{FullAlgo}$ and prove its memory and regret guarantees.

\subsubsection{Regret analysis}
We start with the regret analysis first,  since the memory analysis depends on it.
Formally, we aim to show:

\begin{proposition}[regret bound]
    \label{prop:regret-full}
    For any level $k \in [K]$ and episode $r\in [T/T_{k}]$, the $\textsc{Baseline}_{+}(k)$ has a total regret of $O\left(n^{k+1}\log^{2k} (nT)\right)$ with probability at least $1 - 1/\poly(T)$.
\end{proposition}

We prove the claim by  an induction on $k$. The case $k=1$ follows directly from \cref{prop:regret} by taking $T = T_1 = n^2/\eps$ and $B = 1/\eps^2$. 
Suppose the induction holds up to level $k - 1$ ($k \geq 2$), i.e., 
\begin{align}
  \mathcal{L}_{r, k, t, b}(\textsc{Baseline}_{+}(k-1)) - \mathcal{L}_{r, k, t, b}(i) \leq  \eps^{k-1}\log^{2k-2}(nT). \label{eq:induction-hypothesis}
\end{align}
We proceed for level $k$ and prove the claim for any episode $r \in [T/T_{k}]$.

We first state some basic properties on $\mathcal{L}_{k,r, t, b}(e_{k, i})$ and $\widehat{\mathcal{L}}_{k,r, t, b}(e_{k, i})$. The first claim states $e_{k, i}$ is relatively good on each decision day (since it runs MWU over expert $i$ and $\textsc{Baseline}_{+}(k-1)$), and the second claim states $\widehat{\mathcal{L}}_{k,r, t, b}(e_{k, i})$ has small width.

\begin{lemma}
\label{lem:baseline+k}
For epoch $t \in [\eps n]$ and decision day $b \in [B]$, with probability at least $1- 1/\poly(T)$, one has
\begin{itemize}
    \item $\mathcal{L}_{k,r, t, b}(e_{k, i}) \leq \min\left\{ \mathcal{L}_{k,r, t, b}(i), \mathcal{L}_{k,r, t, b}(\textsc{Baseline}_{+}(k-1)) \right\} + O\left(\sqrt{\log(nT) /T_{k-1}} \right)$; and
    \item $\widehat{\mathcal{L}}_{k,r, t, b}(e_{k, i} ) \in \left[-\eps^{k-1}\log^{2k-1} (nT), \eps^{k-1}\log^{2k-1} (nT)\right]$.
\end{itemize}
\end{lemma}
\begin{proof}
The first claim follows directly from the MWU guarantee of $\textsc{MergeExpert}$, in particular, with probability at least $1 - 1/\poly(T)$,
\begin{align}
\mathcal{L}_{k,r, t, b}(e_{k, i}) \leq &~ \min\left\{ \mathcal{L}_{k,r, t, b}(i), \mathcal{L}_{k,r, t, b}\left(\textsc{Baseline}_{+}(k-1)\right)\right\} + O\left(\sqrt{\log (nT) /T_{k-1}}\right).\label{eq:baseline-claim-1-1}
\end{align}
For the second claim, we have
\begin{align*}
    \mathcal{L}_{k,r, t, b}(e_{k, i}) - \mathcal{L}_{k,r, t, b}(\textsc{Baseline}_{+}(k-1)) \leq &~   \mathcal{L}_{k,r, t, b}(i) -\mathcal{L}_{k,r, t, b}(\textsc{Baseline}_{+}(k-1)) + O\left(\sqrt{\log(nT) /T_{k-1}}\right)\\
    \leq &~ 2\eps^{k-1}\log^{2k-2} (nT).
\end{align*}
where the first step follows from \cref{eq:baseline-claim-1-1}, the second step holds due to induction hypothesis (\cref{eq:induction-hypothesis}) and $T_{k-1} = n(n/\eps)^{k-1}$. Therefore,
\begin{align*}
    \widehat{\mathcal{L}}_{k,r, t, b}(e_{k, i}) = &~ \max\left\{\mathcal{L}_{k,r, t, b}(e_{k, i}) - \mathcal{L}_{k,r, t, b}(\textsc{Baseline}_{+}(k-1)), -\eps^{k-1}\log^{2k-1} (nT)\right\}\\
    \in &~ \left[-\eps^{k-1}\log^{2k-1} (nT)), \eps^{k-1}\log^{2k-1} (nT)\right].
\end{align*}
We finish the proof here.
\end{proof}

Next we show that 
even though we update MWU in $\textsc{Baseline}_{+}(k)$ using the truncated loss,  the regret with respect to the original expert $[n]$ can be still be bounded.
\begin{lemma}
\label{lem:baseline-epoch}
For any epoch $t\in [\eps n]$, suppose $e_{k, i} \in \mathcal{P}_{k,r,t}$, then with probability at least $1-1/\poly(T)$,
\[
\sum_{b=1}^{B}\mathcal{L}_{k, r, t, b}(e_{k, i_{k,r,t, b}}) \leq \sum_{b=1}^{B}\widehat{\mathcal{L}}_{k, r, t, b}(e_{k, i}) + \mathcal{L}_{k,r, t, b}(\textsc{Baseline}_{+}(k-1)) + \frac{1}{4}\eps^{k-2}\log^{2k} (nT).
\]
\end{lemma}
\begin{proof}
By   \cref{lem:baseline+k}, we note that $\widehat{\mathcal{L}}_{k, r, t}(e_{k, j}) \in [-\eps^{k-1}\log^{2k-1} (nT)), \eps^{k-1}\log^{2k-1} (nT)]$ for any $j \in [n]$. Recall that $\textsc{Baseline}_{+}(k)$ runs MWU in epoch $t$ with $B = 1/\eps^2$ decision days, hence, with probability at least $1 - \poly(T)$,
\begin{align}
    \sum_{b=1}^{B}\widehat{\mathcal{L}}_{k, r, t, b}(e_{k, i_{k,r,t, b}}) \leq &~ \sum_{b=1}^{B}\widehat{\mathcal{L}}_{k, r, t, b}(e_{k, i}) + \eps^{k-1}\log^{2k-1} (nT) \cdot O\left(\sqrt{\log(nT)}/\eps\right)\notag \\
    \leq &~ \sum_{b=1}^{B}\widehat{\mathcal{L}}_{k, r, t, b}(e_{k, i}) + \frac{1}{4}\eps^{k-2}\log^{2k} (nT).\label{eq:baseline-epoch1}
\end{align}
The LHS satisfies
\begin{align}
    \sum_{b=1}^{B}\widehat{\mathcal{L}}_{k, r, t, b}(e_{k, i_{k,r,t, b}}) = &~ \sum_{b=1}^{B}\max\left\{\mathcal{L}_{k,r, t, b}(e_{k, i_{k, r,t, b}}) - \mathcal{L}_{k,r, t, b}(\textsc{Baseline}_{+}(k-1)), -\eps^{k-1} \log^{2k-1} (nT)\right\} \notag \\
    \geq &~\mathcal{L}_{k,r, t, b}(e_{k, i_{k, r,t, b}}) - \mathcal{L}_{k,r, t, b}(\textsc{Baseline}_{+}(k-1)).\label{eq:baseline-epoch3}
\end{align}
Combining \cref{eq:baseline-epoch1} and \cref{eq:baseline-epoch3}, we conclude the proof.
\end{proof}

\begin{lemma}
\label{lem:baseline-epoch2}
For any $i \in [n]$ and epoch $t \in [\eps n]$, with probability at least $1-1/\poly(T)$,
\[
\sum_{b=1}^{B}\mathcal{L}_{k, r, t, b}(e_{k, i_{k,r,t, b}}) \leq \sum_{b=1}^{B}\mathcal{L}_{k, r, t, b}(i) +O\left(\eps^{k-3}\log^{2k-2}(nT)\right).
\]
\end{lemma}
\begin{proof}
For any expert $e_{k, i}$, any decision day $b\in [B]$, since $\textsc{MergeExp}$ runs MWU over $\textsc{Baseline}_{+}(k-1)$ and expert $i$, then with probability at least $1-1/\poly(T)$,
\begin{align*}
\mathcal{L}_{k, r, t, b}(e_{k, i}) \leq &~ \mathcal{L}_{k, r, t, b}(\textsc{Baseline}_{+}(k-1)) + O\left(\sqrt{\log(nT)/T_{k-1}}\right) \\
\leq &~ \mathcal{L}_{k, r, t, b}(i) + \eps^{k-1}\log^{2k-2}(nT) + O\left(\sqrt{\log(nT)/T_{k-1}}\right) \\
\leq &~ \mathcal{L}_{k, r, t, b}(i) + O\left(\eps^{k-1}\log^{2k-2}(nT)\right).
\end{align*}
The first step follows from the guarantee of MWU, the second step follows from induction hypothesis (\cref{eq:induction-hypothesis}) and the last step follows from the choice of $T_{k-1}$. Summing over $b\in [B]$, we get the desired bound.
\end{proof}


Let $i_{k, r}^{*} \in [n]$ be the optimal expert in the $r$-th episode.
We use the same stochastic process of \cref{sec:baseline-regret} to define the unlucky epoch $\mathcal{B}_{k, r} \subseteq [\eps n]$ and assign $i(t) \in [n] \cup \{\mathsf{nil}\}$ ($t\in [\eps n]$). \
Note that we fix all the randomness used at level $1,2,\ldots, k-1$ and episode $1, 2, \ldots, r-1$ in advance, and so the stochastic process depends only on the randomness of $\textsc{Baseline}_{+}(k)$ inside episode $r$.

We can similarly bound the  number of unlucky epochs.
\begin{lemma}[unlucky epoch]
\label{lem:unluck-full}
With probability at least $1- 1/\poly(T)$, 
\begin{enumerate}[(i)]
    \item $|\mathcal{B}_{k,r}| \leq O(\eps^{2}n \log T)$ and 
    \item $\sum_{b=1}^{B}\mathcal{L}_{k, r, t, b}(e_{k, i_{k,r,t, b}}) - \sum_{b=1}^{B}\mathcal{L}_{k, r, t, b}( i_{k, r}^{*}) \leq O\left(\eps^{k-3}\log^{2k-2}(nT)\right)$ for any $t \in \mathcal{B}_{k,r}$.
\end{enumerate} 
\end{lemma}
\begin{proof}
The first claim follows directly from \cref{lem:unluck}, the second claim follows from \cref{lem:baseline-epoch2}.
\end{proof}

The covering property holds similarly:
\begin{lemma}[cover]
\label{lem:cover-full}
For any epoch $t \in [\eps n] \backslash \mathcal{B}_{k,r}$, suppose $\beta(\tau) \leq t < \beta(\tau + 1)$, then with probability at least $1-1/\poly(T)$, we have
\begin{enumerate}[(i)]
    \item $i(\beta(\tau)) = \cdots = i(t) = \cdots = i(\beta(\tau+1) -1) = i_{k, r, \tau}^*  \neq \mathsf{nil}$; 
    \item $i_{k, r, \tau}^* \in \mathcal{P}_{k, r, \nu}$ for any $\nu \in [\beta(\tau): \beta(\tau+1) -1]$; {and}
    \item 
    \begin{align*}
    &~\sum_{\nu=\beta(\tau)}^{\beta(\tau +1) -1}\sum_{b=1}^{B}\widehat{\mathcal{L}}_{k,r,t,b}(e_{k, i(t)}) + \mathcal{L}_{k,r, t, b}(\textsc{Baseline}_{+}(k-1))\\
    < &~ \sum_{\nu=\beta(\tau)}^{\beta(\tau +1) -1}\sum_{b=1}^{B}\mathcal{L}_{k,r,t,b}(i_{k,r}^{*}) + 2\eps^{k-2}\log^{2k-1} (nT) (\beta(\tau+1) - \beta(\tau)).
    \end{align*}
\end{enumerate} 
\end{lemma}

\begin{proof}
The first two claims follow exactly from \cref{lem:cover}. For the last claim, first we have
\begin{align}
    \sum_{\nu=\beta(\tau)}^{\beta(\tau +1) -1}\sum_{b=1}^{B}\widehat{\mathcal{L}}_{k,r,t,b}(e_{k, i(t)}) < &~ \sum_{\nu=\beta(\tau)}^{\beta(\tau +1) -1}\sum_{b=1}^{B}\widehat{\mathcal{L}}_{k,r,t,b}(e_{k, i_{k,r}^{*}}) + \eps^{k}\log^{2k-1} (nT) B (\beta(\tau+1) - \beta(\tau))\notag \\
    = &~ \sum_{\nu=\beta(\tau)}^{\beta(\tau +1) -1}\sum_{b=1}^{B}\widehat{\mathcal{L}}_{k,r,t,b}(e_{k, i_{k,r}^{*}}) + \eps^{k-2}\log^{2k-1} (nT) (\beta(\tau+1) - \beta(\tau)), \label{eq:cover-full1}
\end{align}
where the first step follows from the third claim of \cref{lem:cover-full} by replacing $\eps$ with $\eps^{k}\log^{2k-1}(nT)$, and the second step comes from the choice of $B = 1/\eps^2$.

With probability at least $1- 1/\poly(T)$, we have
\begin{align}
    &~ \sum_{b=1}^{B}\widehat{\mathcal{L}}_{k, r, t, b}(e_{k, i_{k,r}^{*}})\notag\\
    = &~ \sum_{b=1}^{B}\max\{\mathcal{L}_{k,r, t, b}(e_{k, i_{k,r}^{*}}) - \mathcal{L}_{k,r, t, b}(\textsc{Baseline}_{+}(k-1)), -\eps^{k-1}\log^{2k-1}(nT)\}\notag\\
    \leq &~ \sum_{b=1}^{B}\max\left\{\mathcal{L}_{k,r,t, b}(i_{k,r}^{*}) - \mathcal{L}_{k,r, t, b}(\textsc{Baseline}_{+}(k-1)) + O\left(\sqrt{\log(nT)/T_{k-1}}\right), -\eps^{k-1}\log^{2k-1}(nT)\right\}\notag\\
    \leq &~ \sum_{b=1}^{B}\mathcal{L}_{k,r,t, b}(i_{k,r}^{*}) - \mathcal{L}_{k,r, t, b}(\textsc{Baseline}_{+}(k-1)) + O(\sqrt{\log (nT)/T_{k-1}})\cdot (1/\eps^2)\notag \\
    \leq &~ \sum_{b=1}^{B}\mathcal{L}_{k,r,t, b}(i_{k,r}^{*}) - \mathcal{L}_{k,r, t, b}(\textsc{Baseline}_{+}(k-1)) + \eps^{k-2}\log^{2k-1} (nT).\label{eq:cover-full2}
\end{align}
The first step follows from the definition, the second step follows from the first claim in \cref{lem:baseline+k}, the third step holds due to the induction hypothesis (\cref{eq:induction-hypothesis}), the last step follows from the choice of $T_{k-1}$.

Combining \cref{eq:cover-full1} and \cref{eq:cover-full2}, we complete the proof.
\end{proof}

We can now wrap up the proof of \cref{prop:regret-full}:

\begin{proof}[Proof of   \cref{prop:regret-full}]

With probability at least $1-1/\poly(T)$, we have
\begin{align*}
    &~ \sum_{t=1}^{\eps n}\sum_{b=1}^{B}\mathcal{L}_{k,r,t, b}\left(e_{k, i_{k,r,t,b}}\right) - \mathcal{L}_{k,r,t, b}(i_{k,r}^{*})\\
    = &~\sum_{t\in[\eps n]\backslash \mathcal{B}_{k,r}}\sum_{b=1}^{B}\left(\mathcal{L}_{k,r,t, b}\left(e_{k, i_{k,r,t,b}}\right) - \mathcal{L}_{k,r,t, b}(i_{k,r}^{*})\right) + \sum_{t\in \mathcal{B}_{k,r}}\sum_{b=1}^{B}\left(\mathcal{L}_{k,r,t, b}\left(e_{k, i_{k,r,t,b}}\right) - \mathcal{L}_{k,r,t, b}(i_{k,r}^{*})\right) \\
    \leq &~ \sum_{t\in[\eps n]\backslash \mathcal{B}_{k,r}}\sum_{b=1}^{B}\left(\mathcal{L}_{k,r,t, b}\left(e_{k,i_{k,r,t,b}}\right) - \mathcal{L}_{k,r,t, b}(i_{k,r}^{*})\right) + O(\eps^2 n\log T) \cdot O(\eps^{k-3}\log^{2k-2} (nT))\\
    \leq &~ \sum_{t\in[\eps n]\backslash \mathcal{B}_{k,r}}\sum_{b=1}^{B}\left(\widehat{\mathcal{L}}_{k,r,t, b}(e_{k, i(t)}) + \mathcal{L}_{k,r,t, b}(\textsc{Baseline}_{+}(k-1)) - \mathcal{L}_{k,r,t, b}(i_{k,r}^{*})\right)\\
    &~ + \frac{1}{4}\eps^{k-2}\log^{2k} (nT)\cdot \eps n + O(\eps^{k-1} n \log^{2k-1} (nT))\\
    = &~ \sum_{\tau}\sum_{t \in [\beta(\tau): \beta(\tau + 1) -1], t\notin \mathcal{B}}\sum_{b=1}^{B}\left(\widehat{\mathcal{L}}_{k,r,t, b}(e_{k, i(t)}) + \mathcal{L}_{k,r,t, b}(\textsc{Baseline}_{+}(k-1)) - \mathcal{L}_{k,r,t, b}(i_{k,r}^{*})\right)\\
    &~ +\frac{1}{2}\eps^{k-1}n\log^{2k}(nT) \\
    \leq &~ \sum_{\tau} 2\eps^{k-2} \log^{2k-1}(nT) (\beta(\tau+1) - \beta(\tau)) + \frac{1}{2}\eps^{k-1}n\log^{2k}(nT) \leq \eps^{k-1}n\log^{2k} (nT).
\end{align*}
The second step follows from \cref{lem:unluck-full}, the third step is via \cref{lem:baseline-epoch} and the fact that $i(t) \in \mathcal{P}_{k, r, t}$ for any $t \in [\eps n]\backslash \mathcal{B}_{k,r}$ (see the second claim of \cref{lem:cover-full}), the fourth step follows from \cref{lem:cover-full}, and the fifth step   from the third claim of \cref{lem:cover-full}.

Hence, the average regret of the $r$-th episode equals $\eps^{k-1}n\log^{2k} (nT)/(\eps n B) = \eps^{k}n\log^{2k} (nT)$. We finish the induction and complete the proof here.
\end{proof}

\subsubsection{Memory}

We bound the memory requirement of $\textsc{Baseline}_{+}(k)$ (for each $k \in [K]$).
The proof essentially inherits from \cref{sec:memory-baseline}, with the key observation that $\widehat{\mathcal{L}}_{k,r,t, b} \in [-\eps^{k-1}\log^{2k-1}(nT), \eps^{k-1}\log^{2k-1}(nT)]$ instead of $[0, 1]^{n}$. 
This allows one to perform a fine-grained division like \cref{eqn:new-evict-rule}. In the remaining of this section, we always condition on the high probability event of the previous section.

The following lemma is similar to \cref{lem:loss-length}.
\begin{lemma}\label{lem:loss-length-full}
    For any level $k \in [2:K]$, episode $r\in T_{k}$ and epoch $t \in [\eps n]$, suppose experts $e_{k, i}, e_{k, j}\in \widetilde{\mathcal{P}}_{k, r, t}$ and $e_{k, j}\succ e_{k, i}$.  Let $\alpha \in (0, 1)$, then at least one of the following must hold:
    \begin{enumerate}[(i)]
        \item $\widehat{\mathcal{L}}_{j,j} \geq  \widehat{\mathcal{L}}_{i,i} + 2\eps^{k-1}\log^{2k-1}(nT)(\eps/2 - \alpha)$;
        \item $|\Gamma_{k, r, t, j}|  \geq \left(1 + \frac{\alpha}{1 - \alpha}\right)|\Gamma_{k,r,t,i}|$.
    \end{enumerate}
    Here $\widehat{\mathcal{L}}_{i,i} = \widehat{\mathcal{L}}_{\Gamma_{k, r, t,i} }(e_{k, i}) $, $\widehat{\mathcal{L}}_{i,j} =\mathcal{L}_{\Gamma_{k,r,t, i} }(e_{k, j}) $ and accordingly $\widehat{\mathcal{L}}_{j,j} =\mathcal{L}_{\Gamma_{k, r, t, j} }(e_{k, j})$
\end{lemma}
\begin{proof}
Normalizing $\widehat{\mathcal{L}}_{k,r,t, b}$ by a factor of  $2\eps^{k-1}\log^{2k-1}(nT)$, by \cref{lem:baseline+k}, we know that the truncated loss $\widehat{\mathcal{L}}_{k,r,t,b}/2\eps^{k-1}\log^{2k-1}(nT) \in [-1/2,1/2]^n$. The new eviction rule (\cref{eqn:new-evict-rule})   reduces to the old one (\cref{eqn:evict-rule}) with $\eps$ replaced by $\eps/2$. Then we can apply \cref{lem:loss-length} and get the desired.
\end{proof}

Similar to \cref{lem:pool-size},  it follows that  the pool size is small. The proof is analogous to \cref{lem:pool-size}, by applying the same potential function using the conditions from \cref{lem:loss-length-full}.
\begin{lemma}[pool size]
\label{lem:pool-size-full}
   For any level $k \in [K]$, episode $r\in T_{k}$ and epoch $t \in [\eps n]$, the size of the pool $\widetilde{P}_{k,r, t}$ is at most $S_{k} =  O(\eps^{-1}\log T)$.
\end{lemma}

Now we can wrap up the memory requirement.
\begin{proposition}[memory bound]
    \label{prop:memory-full}
    At any time during the execution of \textsc{FullAlgo}, the memory usage is at most  $O\left(\frac{1}{\eps^2} \log^4 (nT)\right)$ bits.
\end{proposition}
\begin{proof}
At any time step, \textsc{FullAlgo} maintains $\textsc{Baseline}_{+}(k)$ for each $k \in K$ and experts $\{e_{i, k}\}_{i \in [n], k \in [K]}$. For $\textsc{Baseline}_{+}(k)$, by \cref{lem:pool-size-full}, the size of the pool never exceeds $O(\eps^{-1}\log T)$ and therefore it takes $O(\eps^{-2}\log^3 nT)\cdot K = O(\eps^{-2}\log^4 nT)$ bits of memory in total. Note that $\textsc{Baseline}_{+}(k)$ tracks $\{e_{i, k}\}_{i \in [n]}$ instead of the original expert. This does not take extra memory, since we always maintain $\textsc{Baseline}_{+}(k-1)$ and perform MWU on $\textsc{Baseline}_{+}(k-1)$ and expert $i$ does not take extra memory.
\end{proof}

Combining \cref{prop:regret-full} and \cref{prop:memory-full}, we can prove \cref{thm:sublinear}.

\begin{proof}[Proof of \cref{thm:sublinear}]
Taking $\eps^{-1} = n^{-\delta/2}$, by \cref{prop:regret-full}, the memory never exceeds $O(n^{\delta} \log^4(nT))$.
For regret analysis, suppose $n(n/\eps)^{K}\leq T < n(n/\eps)^{K+1}$ for some integer $K\geq 0$.
The \cref{prop:regret-full} states that within $n(n/\eps)^{K+1}$ days, the total regret is at most $O\left(n^{K + 2}\log^{2K+2}(nT)\right) \leq \widetilde{O}\left(n^2 T^{\frac{2}{2+\delta}}\right)$. We conclude the proof here.
\end{proof}

\begin{remark}
Our results extend easily to the case that $T$ is unknown in advance: One can apply the common doubling trick and obtain the same result.
\end{remark}

\section{Lower bound against adaptive adversary}
\label{sec:lower}

We prove no algorithm can achieve sub-linear regret using sub-linear space when facing an adaptive adversary (see  \cref{def:adaptive}).
\begin{theorem}[Lower bound against adaptive adversary]
\label{thm:lower}
Let $n, T > 0$ be sufficiently large, $0 < \eps < 1/40$. Any algorithm that achieves $O(\eps T)$ regret against an adaptive adversary requires at least $\Omega(\min\{\eps^{-1}\log_2 n, n\})$ bits of memory.
\end{theorem}

Our lower bound construction utilizes the well-established connection of no-regret learning and zero-sum games. We first construct a family of zero-sum games whose equilibria are far apart (\cref{lem:minmax} and  \cref{lem:game-loss}). We then prove it serves as a hard distribution for the online learning task. Throughout the proof, we assume $\eps^{-1} \leq \frac{n}{20\log_2 n}$ and $k = 1/(2\eps)$ is an integer.

\paragraph{Hard distribution} We construct a family of zero-sum games. 
For any set $S \subseteq [n]$ of size $k$, the game matrix $A_{S} \in [0,4]^{n \times n}$ determines the loss of the first player (Alice):
\begin{align*}
    A_{S}[i, j] = \left\{
    \begin{matrix}
    4 & i \notin S\\
    0 & i \in S, i \neq j\\
    1 & i \in S, i = j.
    \end{matrix}
    \right.
\end{align*}
To summarize, Alice receives loss $4$ if she plays any action outside of the support of $S$, and they are strictly dominated by actions in $S$. For any action pair $(i, j) \in S \times S$, the game is constructed as a generalized matching penny game: Alice receives loss $1$ if her action is matched by Bob, and she receives loss $0$ otherwise.
The (hard) distribution $\D_{k}$ is defined as the uniform distribution over the above family of zero-sum games $\{A_{S}\}_{S \subseteq [n], |S| = k}$.

We first make a simple observation about the equilibrium strategy.
\begin{lemma}\label{lem:minmax}
For any $S \subseteq [n]$ of size $k$, the minmax value of game $A_S$ equals $1/k$. Furthermore, Alice's equilibrium strategy is $\frac{1}{k}\cdot \mathbf{1}_{S}$, where $\mathbf{1}_{S}$ is the indicator vector whose $i$-th entry equals $\mathbf{1}\{i \in S\}$.
\end{lemma}
\begin{proof}
It is easy to verify that $(\mathbf{1}_{S}, \mathbf{1}_{S})$ is the unique equilibrium of the game, and by definition, Alice receives $\frac{1}{k}$ loss in the equilibrium.
\end{proof}

We then prove that no single strategy can (approximately) cover the minmax strategy of a large number of games. 
For any strategy $p \in \Delta_n$ and game $A_S$, define $\ell(p, S)$ as the worst case loss received by Alice when playing $p$ in game $A_S$, i.e., $\ell(p, S)= \max_{i\in [n]}p^{\top}A_S \mathbf{1}_{i}$.
\begin{lemma}\label{lem:game-loss}
For any fixed strategy $p \in \Delta_{n}$, there are at most $\binom{n}{3k/4}$ number of sets $S \subseteq [n]$ ($|S| = k$) such that $\ell(p, S) < 2/k$.
\end{lemma}
\begin{proof}
Without loss of generality, we assume $p_1 \geq p_2 \geq \cdots \geq p_{n}$. 

\textbf{Case 1.} Suppose $p_1 \geq 2/k$. Then for any set $S \subseteq [n]$ ($|S| = k$), there are two cases: either (1) if $1 \in S$, then Bob plays action $1$ and Alice receives at least $p_1 \geq 2/k$ loss; or (2) if $1 \notin S$, then Alice receives at least $4 p_1 \geq 8/k$ loss. Therefore, in this case, there is no $S$ satisfies $\ell(p, S) \leq 2/k$.

\textbf{Case 2.} Suppose $p_1 < 2/k$. Let $i^{*} \in [n]$ be the largest integer such that $p_{i^{*}} \geq 1/2k$ and let $I = [i^{*}]$. Note if $p_1 < 1/2k$,   we simply take $I = \emptyset$. 

For any $S$ satisfies $\ell(p, S) < 2/k$, we first prove $I \subseteq S$. Otherwise, suppose that there exists an index $i \in I$ such that $i \notin S$, then by having Bob play action $i$, Alice receives at least $4p_i \geq 2/k$ loss.  

\textbf{Case 2-1.} Suppose $|I| < k/4$, then we claim that there is no set $S$ satisfies $\ell(p, S) < 2/k$. To see this, for any set $S\subseteq [n]$ of size $k$, we note that 
\[
\sum_{i \in S} p_i = \sum_{i \in I}p_i + \sum_{i \in S\backslash I}p_i \leq \frac{k}{4}\cdot \frac{2}{k} + \frac{3k}{4} \cdot\frac{1}{2k}  = \frac{7}{8} ,
\]
where the second step holds as $|I| < k/4$, and for $i \in I, p_i < p_1 = 2/k$, for $i \in S\backslash I$, $p_i < 1/2k$. Hence, we have $\sum_{i \notin S} p_i \geq 1/8$ and Alice receives at least $1/8 \cdot 4 = 1/2$ loss.

\textbf{Case 2-2.} Suppose $|I| > 1/4k$. We have already proved $I \subseteq S$, and therefore, there are at most $3k/4$ indices in $S\backslash I$ and they can be chosen from $[n]\backslash I$, which can be upper bounded by $\binom{n}{3k/4}$.

We conclude the lemma here.
\end{proof}

Now we can prove the lower bound against adaptive adversary. The high-level idea is that one can use a lower memory no-regret algorithm to approximate the minmax value of a zero-sum game (where the opponent always plays the best response). Meanwhile, by a counting argument, the number of different zero sum games constructed above are too large and can not be ``covered'' by a low memory algorithm.

\begin{proof}[Proof of   \cref{thm:lower}]
Let $\ALG$ be any algorithm with asymptotic regret of $O(\eps T)$ when facing a strong adaptive adversary (\cref{def:adaptive}). Suppose $\ALG$ uses $M$ bits of memory and we shall prove $M \geq \Omega\left(\eps^{-1}\log_2 n\right)$ (note we assume $\eps^{-1} \leq \frac{n}{\log_2 n}$ at the very beginning).
Consider the following scenario.

\begin{itemize}
    \item A game  $A_{S}$  is sampled from the distribution $\D_{k}$.
    \item Alice and Bob play the game $A_S$ repeatedly for $T$ rounds. In the $t$-th round ($t\in [T]$)
    \begin{itemize}
        \item Alice consults with $\ALG$ and commits a distribution $p_t \in \Delta_{n}$ over her $n$ actions;
        \item Bob receives $p_t$ and best responds to Alice with $y_t = \arg\max_{i \in [n]} \left(p_t^{\top}A\right)_{i}$.
    \end{itemize}
\end{itemize}

Slightly abuse of notation, we would also use $y_t$ as an indicator vector.
At the end, Alice would achieve $\eps$-approximate to its minmax value:
\begin{align}
    \frac{1}{T}\sum_{t=1}^{T} \langle p_t, A_S y_t\rangle \leq  \frac{1}{T}\min_{i^*}\sum_{t=1}^{T} \langle i^*, A_S y_t\rangle + \eps
    \leq \frac{1}{k} + \eps \leq \frac{3}{2k}. \label{eq:minmax}
\end{align}
The first step follows from the regret guarantee of $\ALG$, the second step follows from the minmax Theorem (\cref{lem:minmax-thm}) and the minmax value of $A_S$ always equals to $1/k$ (\cref{lem:minmax}), the last step follows from the choice of $k$.

We prove by contradiction and assume that $\ALG$ uses at most $M = \frac{1}{10}\eps^{-1}\log_2 n$ bits of memory.
We aim to prove that Alice has loss at least $3/2k$ for every iteration with high probability (over the choice of $S$). The proof is via the following counting argument. 

As the algorithm uses at most $\frac{1}{10}\eps^{-1}\log_2 n$ bits of memory, there are $2^{M} = n^{k/5}$ possible memory states $X$ in total. 
For each state $x \in X$, suppose that the algorithm outputs a strategy $p_x \in \Delta_n$ (note $p_x$ could be a random variable). Define 
\begin{align*}
T_x := \left\{S: \Pr[\ell(p_x, S) < 2/k] \geq 0.1, S\in [n], |S| = k       \right\}.
\end{align*}

By  \cref{lem:game-loss}, we know that $|T_x| \leq 10\cdot \binom{n}{3k/4}$. Taking an union over all states $x\in X$, one has
\[
\left|\bigcup_{x\in X} T_x\right| \leq n^{k/5} \cdot 10\binom{n}{3k/4} \leq \frac{1}{100}\binom{n}{k}.
\]
The last step follows from the choice of parameters.

Hence, we conclude that with probability at least $0.99$, the nature draws a game $A_S$ such that $S \notin \bigcup_{x\in X} T_x$.
This means that in each round of the game, the algorithm receives at least $\frac{1}{10}\frac{1}{k} + \frac{9}{10}\frac{2}{k} = \frac{19}{10k}$ loss in expectation since Bob always plays best response.
 This contradicts with \cref{eq:minmax} and we conclude the proof.
\end{proof}
\section{Conclusion} \label{sec:con}
In this paper, we provide the first sub-linear space online learning algorithm that achieves sub-linear regret when facing an oblivious adversary. 
A separation has also been established between oblivious and strong adaptive adversaries, where a linear memory lower bound is shown to be necessary to achieve sub-linear regret.

Our work opens up a variety of exciting future research directions:
\begin{enumerate}[(i)]
    \item First, an immediately interesting question is to close the gap between the space upper and lower bound.
    \item Second, the adaptive adversary model considered in the paper is relatively strong for some applications, where the adversary only sees the prior decisions but not the mixed strategy of current round. A natural open question is to investigate this model, often called black-box adversary in the adversarially robust streaming literature.
    \item Third, other problem-specific notions of  regret have been studied under different settings, such as dynamic environments   \cite{foster2015adaptive} and game theory \cite{blum2007external,chen2020hedging,anagnostides2022near}. A natural follow-up question is whether sub-linear space is achievable there.
    \item Finally, the MWU algorithm has numerous applications in game theory  \cite{freund1999adaptive} and machine learning (including boosting \cite{freund1997decision} and reinforcement learning \cite{daskalakis2022complexity}).
Our paper opens up opportunities of deriving sub-linear space algorithms for these applications. 
\end{enumerate}

\section*{Acknowledgement}
B.P. and F.Z wish to thank David Woodruff and Samson Zhou for insightful discussions over the project and thank Xi Chen, Jelani Nelson, Christos Papadimitriou, Aviad Rubinstein for helpful comments on early drafts of this paper. 

Fred Zhang is supported by ONR DORECG award N00014-17-1-2127.
Binghui Peng is supported by NSF CCF-1909756, CCF2007443, CCF-2134105, CCF-1703925, IIS-1838154, CCF-2106429 and CCF-2107187.

\clearpage
\bibliographystyle{alpha}
\bibliography{bib}

\newcommand{\etalchar}[1]{$^{#1}$}
\begin{thebibliography}{BEJWY22}

\bibitem[ABED{\etalchar{+}}21]{alon2021adversarial}
Noga Alon, Omri Ben-Eliezer, Yuval Dagan, Shay Moran, Moni Naor, and Eylon
  Yogev.
\newblock Adversarial laws of large numbers and optimal regret in online
  classification.
\newblock In {\em Proceedings of the 53rd Annual ACM SIGACT Symposium on Theory
  of Computing (STOC)}, 2021.

\bibitem[ABIS19]{acharya2019estimating}
Jayadev Acharya, Sourbh Bhadane, Piotr Indyk, and Ziteng Sun.
\newblock Estimating entropy of distributions in constant space.
\newblock {\em Advances in Neural Information Processing Systems (NeurIPS)},
  2019.

\bibitem[ABJ{\etalchar{+}}22]{10.1145/3517804.3526228}
Mikl\'{o}s Ajtai, Vladimir Braverman, T.S. Jayram, Sandeep Silwal, Alec Sun,
  David~P. Woodruff, and Samson Zhou.
\newblock The white-box adversarial data stream model.
\newblock In {\em Proceedings of the 41st ACM SIGMOD-SIGACT-SIGAI Symposium on
  Principles of Database Systems (PODS)}, 2022.

\bibitem[ADF{\etalchar{+}}22]{anagnostides2022near}
Ioannis Anagnostides, Constantinos Daskalakis, Gabriele Farina, Maxwell
  Fishelson, Noah Golowich, and Tuomas Sandholm.
\newblock Near-optimal no-regret learning for correlated equilibria in
  multi-player general-sum games.
\newblock In {\em Proceedings of the 54th Annual ACM SIGACT Symposium on Theory
  of Computing (STOC)}, 2022.

\bibitem[AEMP22]{ahmadian2022robust}
Sara Ahmadian, Hossein Esfandiari, Vahab Mirrokni, and Binghui Peng.
\newblock Robust load balancing with machine learned advice.
\newblock In {\em Proceedings of the 2022 Annual ACM-SIAM Symposium on Discrete
  Algorithms (SODA)}, 2022.

\bibitem[AHK12]{arora2012multiplicative}
Sanjeev Arora, Elad Hazan, and Satyen Kale.
\newblock The multiplicative weights update method: a meta-algorithm and
  applications.
\newblock {\em Theory of Computing}, 8(1):121--164, 2012.

\bibitem[AKP22]{agarwal2022sharp}
Arpit Agarwal, Sanjeev Khanna, and Prathamesh Patil.
\newblock A sharp memory-regret trade-off for multi-pass streaming bandits.
\newblock {\em Conference on Learning Theory (COLT)}, 2022.

\bibitem[AMNW22]{aliakbarpour2022estimation}
Maryam Aliakbarpour, Andrew McGregor, Jelani Nelson, and Erik Waingarten.
\newblock Estimation of entropy in constant space with improved sample
  complexity.
\newblock In {\em Advances in Neural Information Processing Systems (NeurIPS)},
  2022.

\bibitem[AMS99]{alon1999space}
Noga Alon, Yossi Matias, and Mario Szegedy.
\newblock The space complexity of approximating the frequency moments.
\newblock {\em Journal of Computer and System Sciences}, 58(1):137--147, 1999.

\bibitem[AW20]{assadi2020exploration}
Sepehr Assadi and Chen Wang.
\newblock Exploration with limited memory: streaming algorithms for coin
  tossing, noisy comparisons, and multi-armed bandits.
\newblock In {\em Proceedings of the 52nd Annual ACM SIGACT Symposium on Theory
  of Computing (STOC)}, 2020.

\bibitem[BBS22]{brown2022strong}
Gavin Brown, Mark Bun, and Adam Smith.
\newblock Strong memory lower bounds for learning natural models.
\newblock {\em Conference on Learning Theory (COLT)}, 2022.

\bibitem[BDGR22]{noah}
Adam Block, Yuval Dagan, Noah Golowich, and Alexander Rakhlin.
\newblock Smoothed online learning is as easy as statistical learning.
\newblock In {\em Conference on Learning Theory (COLT)}, 2022.

\bibitem[BEJWY22]{ben2022framework}
Omri Ben-Eliezer, Rajesh Jayaram, David~P Woodruff, and Eylon Yogev.
\newblock A framework for adversarially robust streaming algorithms.
\newblock {\em Journal of the ACM (JACM)}, 69(2):1--33, 2022.

\bibitem[BM07]{blum2007external}
Avrim Blum and Yishay Mansour.
\newblock From external to internal regret.
\newblock {\em Journal of Machine Learning Research}, 8(6), 2007.

\bibitem[Bro51]{brown1951iterative}
George~W Brown.
\newblock Iterative solution of games by fictitious play.
\newblock {\em Act. Anal. Prod Allocation}, 13(1):374, 1951.

\bibitem[CBL06]{cesa2006prediction}
Nicolo Cesa-Bianchi and G{\'a}bor Lugosi.
\newblock {\em Prediction, learning, and games}.
\newblock Cambridge university press, 2006.

\bibitem[CDS13]{DBLP:conf/nips/Cesa-BianchiDS13}
Nicol{\`{o}} Cesa{-}Bianchi, Ofer Dekel, and Ohad Shamir.
\newblock Online learning with switching costs and other adaptive adversaries.
\newblock In {\em Advances in Neural Information Processing Systems (NeurIPS)},
  2013.

\bibitem[CK20]{chaudhuri2020regret}
Arghya~Roy Chaudhuri and Shivaram Kalyanakrishnan.
\newblock Regret minimisation in multi-armed bandits using bounded arm memory.
\newblock In {\em Proceedings of the AAAI Conference on Artificial Intelligence
  (AAAI)}, 2020.

\bibitem[CKM{\etalchar{+}}11]{christiano2011electrical}
Paul Christiano, Jonathan~A Kelner, Aleksander Madry, Daniel~A Spielman, and
  Shang-Hua Teng.
\newblock Electrical flows, laplacian systems, and faster approximation of
  maximum flow in undirected graphs.
\newblock In {\em Proceedings of the Forty-third Annual ACM Symposium on Theory
  of Computing (STOC)}, pages 273--282, 2011.

\bibitem[CMVW16]{crouch2016stochastic}
Michael Crouch, Andrew McGregor, Gregory Valiant, and David~P Woodruff.
\newblock Stochastic streams: Sample complexity vs. space complexity.
\newblock In {\em 24th Annual European Symposium on Algorithms (ESA)}, 2016.

\bibitem[CP20]{chen2020hedging}
Xi~Chen and Binghui Peng.
\newblock Hedging in games: Faster convergence of external and swap regrets.
\newblock In {\em Advances in Neural Information Processing Systems (NeurIPS)},
  2020.

\bibitem[CPP22]{chen2022memory}
Xi~Chen, Christos Papadimitriou, and Binghui Peng.
\newblock Memory bounds for continual learning.
\newblock In {\em 2022 IEEE 63th Annual Symposium on Foundations of Computer
  Science (FOCS)}, 2022.

\bibitem[DFT13]{doyle2013feedback}
John~C Doyle, Bruce~A Francis, and Allen~R Tannenbaum.
\newblock {\em Feedback control theory}.
\newblock Courier Corporation, 2013.

\bibitem[DGZ22]{daskalakis2022complexity}
Constantinos Daskalakis, Noah Golowich, and Kaiqing Zhang.
\newblock The complexity of {M}arkov equilibrium in stochastic games.
\newblock {\em arXiv preprint arXiv:2204.03991}, 2022.

\bibitem[DHL{\etalchar{+}}20]{dudik2020oracle}
Miroslav Dud{\'\i}k, Nika Haghtalab, Haipeng Luo, Robert~E Schapire, Vasilis
  Syrgkanis, and Jennifer~Wortman Vaughan.
\newblock Oracle-efficient online learning and auction design.
\newblock {\em Journal of the ACM (JACM)}, 67(5):1--57, 2020.

\bibitem[DKPP22]{diakonikolas2022streaming}
Ilias Diakonikolas, Daniel~M Kane, Ankit Pensia, and Thanasis Pittas.
\newblock Streaming algorithms for high-dimensional robust statistics.
\newblock In {\em International Conference on Machine Learning (ICML)}, 2022.

\bibitem[DLC22]{dai2022follow}
Yan Dai, Haipeng Luo, and Liyu Chen.
\newblock Follow-the-perturbed-leader for adversarial markov decision processes
  with bandit feedback.
\newblock {\em Advances in Neural Information Processing Systems (NeurIPS)},
  2022.

\bibitem[DS18]{dagan2018detecting}
Yuval Dagan and Ohad Shamir.
\newblock Detecting correlations with little memory and communication.
\newblock In {\em Conference On Learning Theory (COLT)}, 2018.

\bibitem[DTA12]{DBLP:conf/icml/DekelTA12}
Ofer Dekel, Ambuj Tewari, and Raman Arora.
\newblock Online bandit learning against an adaptive adversary: from regret to
  policy regret.
\newblock In {\em Proceedings of the 29th International Conference on Machine
  Learning (ICML)}, 2012.

\bibitem[FRS15]{foster2015adaptive}
Dylan~J Foster, Alexander Rakhlin, and Karthik Sridharan.
\newblock Adaptive online learning.
\newblock {\em Advances in Neural Information Processing Systems (NIPS)}, 2015.

\bibitem[FS97]{freund1997decision}
Yoav Freund and Robert~E Schapire.
\newblock A decision-theoretic generalization of on-line learning and an
  application to boosting.
\newblock {\em Journal of Computer and System Sciences}, 55(1):119--139, 1997.

\bibitem[FS99]{freund1999adaptive}
Yoav Freund and Robert~E Schapire.
\newblock Adaptive game playing using multiplicative weights.
\newblock {\em Games and Economic Behavior}, 29(1-2):79--103, 1999.

\bibitem[GH16]{garber2016sublinear}
Dan Garber and Elad Hazan.
\newblock Sublinear time algorithms for approximate semidefinite programming.
\newblock {\em Mathematical Programming}, 158(1):329--361, 2016.

\bibitem[GKLR21]{garg2021memory}
Sumegha Garg, Pravesh~Kumar Kothari, Pengda Liu, and Ran Raz.
\newblock Memory-sample lower bounds for learning parity with noise.
\newblock In {\em 24th International Conference on Approximation Algorithms for
  Combinatorial Optimization Problems (APPROX) and 25th International
  Conference on Randomization and Computation (RANDOM)}, 2021.

\bibitem[GLM20]{gonen2020towards}
Alon Gonen, Shachar Lovett, and Michal Moshkovitz.
\newblock Towards a combinatorial characterization of bounded-memory learning.
\newblock {\em Advances in Neural Information Processing Systems (NeurIPS)},
  2020.

\bibitem[GRT18]{garg2018extractor}
Sumegha Garg, Ran Raz, and Avishay Tal.
\newblock Extractor-based time-space lower bounds for learning.
\newblock In {\em Proceedings of the 50th Annual ACM SIGACT Symposium on Theory
  of Computing (STOC)}, 2018.

\bibitem[Haz16]{hazan2016introduction}
Elad Hazan.
\newblock Introduction to online convex optimization.
\newblock {\em Foundations and Trends{\textregistered} in Optimization},
  2(3-4):157--325, 2016.

\bibitem[HHSY22]{nika}
Nika Haghtalab, Yanjun Han, Abhishek Shetty, and Kunhe Yang.
\newblock Oracle-efficient online learning for beyond worst-case adversaries.
\newblock {\em Advances in Neural Information Processing Systems (NeurIPS)},
  2022.

\bibitem[HK16]{hazan2016computational}
Elad Hazan and Tomer Koren.
\newblock The computational power of optimization in online learning.
\newblock In {\em Proceedings of the Forty-eighth Annual ACM Symposium on
  Theory of Computing (STOC)}, pages 128--141, 2016.

\bibitem[HLZ20]{hopkins2020robust}
Sam Hopkins, Jerry Li, and Fred Zhang.
\newblock Robust and heavy-tailed mean estimation made simple, via regret
  minimization.
\newblock In {\em Advances in Neural Information Processing Systems (NeurIPS)},
  2020.

\bibitem[HRS22]{tim}
Nika Haghtalab, Tim Roughgarden, and Abhishek Shetty.
\newblock Smoothed analysis with adaptive adversaries.
\newblock In {\em 2021 IEEE 62nd Annual Symposium on Foundations of Computer
  Science (FOCS)}, 2022.

\bibitem[HS97]{helmbold1997predicting}
David~P Helmbold and Robert~E Schapire.
\newblock Predicting nearly as well as the best pruning of a decision tree.
\newblock {\em Machine Learning}, 27(1):51--68, 1997.

\bibitem[JHTX21]{jin2021optimal}
Tianyuan Jin, Keke Huang, Jing Tang, and Xiaokui Xiao.
\newblock Optimal streaming algorithms for multi-armed bandits.
\newblock In {\em International Conference on Machine Learning (ICML)}, 2021.

\bibitem[KM17]{klivans2017learning}
Adam Klivans and Raghu Meka.
\newblock Learning graphical models using multiplicative weights.
\newblock In {\em IEEE 58th Annual Symposium on Foundations of Computer Science
  (FOCS)}, 2017.

\bibitem[KSJK13]{on2013}
Purushottam Kar, Bharath~K Sriperumbudur, Prateek Jain, and Harish~C Karnick.
\newblock On the generalization ability of online learning algorithms for
  pairwise loss functions.
\newblock In {\em International Conference on Machine Learning (ICML)}, 2013.

\bibitem[KV05]{kalai2005efficient}
Adam Kalai and Santosh Vempala.
\newblock Efficient algorithms for online decision problems.
\newblock {\em Journal of Computer and System Sciences}, 71(3):291--307, 2005.

\bibitem[KVV90]{karp1990optimal}
Richard~M Karp, Umesh~V Vazirani, and Vijay~V Vazirani.
\newblock An optimal algorithm for on-line bipartite matching.
\newblock In {\em Proceedings of the Twenty-second Annual ACM Symposium on
  Theory of Computing (STOC)}, 1990.

\bibitem[LS20]{lattimore2020bandit}
Tor Lattimore and Csaba Szepesv{\'a}ri.
\newblock {\em Bandit algorithms}.
\newblock Cambridge University Press, 2020.

\bibitem[LSPY18]{liau2018stochastic}
David Liau, Zhao Song, Eric Price, and Ger Yang.
\newblock Stochastic multi-armed bandits in constant space.
\newblock In {\em International Conference on Artificial Intelligence and
  Statistics (AISTATS)}, 2018.

\bibitem[LW89]{littlestone1989weighted}
N~Littlestone and MK~Warmuth.
\newblock The weighted majority algorithm.
\newblock In {\em 30th Annual Symposium on Foundations of Computer Science
  (FOCS)}, 1989.

\bibitem[MPK21]{maiti2021multi}
Arnab Maiti, Vishakha Patil, and Arindam Khan.
\newblock Multi-armed bandits with bounded arm-memory: Near-optimal guarantees
  for best-arm identification and regret minimization.
\newblock {\em Advances in Neural Information Processing Systems (NeurIPS)},
  2021.

\bibitem[MR95]{motwani1995randomized}
Rajeev Motwani and Prabhakar Raghavan.
\newblock {\em Randomized algorithms}.
\newblock Cambridge University Press, 1995.

\bibitem[MSSV22]{marsden2022efficient}
Annie Marsden, Vatsal Sharan, Aaron Sidford, and Gregory Valiant.
\newblock Efficient convex optimization requires superlinear memory.
\newblock In {\em Conference on Learning Theory (COLT)}, 2022.

\bibitem[MW98]{maass1998efficient}
Wolfgang Maass and Manfred~K Warmuth.
\newblock Efficient learning with virtual threshold gates.
\newblock {\em Information and Computation}, 141(1):66--83, 1998.

\bibitem[Neu28]{neumann1928theorie}
J.~v. Neumann.
\newblock Zur theorie der gesellschaftsspiele.
\newblock {\em Mathematische Annalen}, 100(1):295--320, 1928.

\bibitem[OC98]{ordentlich1998cost}
Erik Ordentlich and Thomas~M Cover.
\newblock The cost of achieving the best portfolio in hindsight.
\newblock {\em Mathematics of Operations Research}, 23(4):960--982, 1998.

\bibitem[Raz17]{raz2017time}
Ran Raz.
\newblock A time-space lower bound for a large class of learning problems.
\newblock In {\em 2017 IEEE 58th Annual Symposium on Foundations of Computer
  Science (FOCS)}, 2017.

\bibitem[Raz18]{raz2018fast}
Ran Raz.
\newblock Fast learning requires good memory: A time-space lower bound for
  parity learning.
\newblock {\em Journal of the ACM (JACM)}, 66(1):1--18, 2018.

\bibitem[RST11]{NIPS2011_692f93be}
Alexander Rakhlin, Karthik Sridharan, and Ambuj Tewari.
\newblock Online learning: Stochastic, constrained, and smoothed adversaries.
\newblock In {\em Advances in Neural Information Processing Systems (NIPS)},
  2011.

\bibitem[She14]{sherstov2014communication}
Alexander~A Sherstov.
\newblock Communication complexity theory: Thirty-five years of set
  disjointness.
\newblock In {\em International Symposium on Mathematical Foundations of
  Computer Science (MFCS)}, 2014.

\bibitem[Sli19]{slivkins2019introduction}
Aleksandrs Slivkins.
\newblock Introduction to multi-armed bandits.
\newblock {\em Foundations and Trends{\textregistered} in Machine Learning},
  12(1-2):1--286, 2019.

\bibitem[SSV19]{sharan2019memory}
Vatsal Sharan, Aaron Sidford, and Gregory Valiant.
\newblock Memory-sample tradeoffs for linear regression with small error.
\newblock In {\em Proceedings of the 51st Annual ACM SIGACT Symposium on Theory
  of Computing (STOC)}, 2019.

\bibitem[SVW16]{steinhardt2016memory}
Jacob Steinhardt, Gregory Valiant, and Stefan Wager.
\newblock Memory, communication, and statistical queries.
\newblock In {\em Conference on Learning Theory (COLT)}, 2016.

\bibitem[SWXZ22]{stoc22}
Vaidehi Srinivas, David~P. Woodruff, Ziyu Xu, and Samson Zhou.
\newblock Memory bounds for the experts problem.
\newblock In {\em Proceedings of the 54th Annual ACM SIGACT Symposium on Theory
  of Computing (STOC)}, 2022.

\bibitem[TMV01]{takimoto2001predicting}
Eiji Takimoto, Akira Maruoka, and Volodya Vovk.
\newblock Predicting nearly as well as the best pruning of a decision tree
  through dynamic programming scheme.
\newblock {\em Theoretical Computer Science}, 261(1):179--209, 2001.

\bibitem[Waj20]{wajc2020rounding}
David Wajc.
\newblock Rounding dynamic matchings against an adaptive adversary.
\newblock In {\em Proceedings of the 52nd Annual ACM SIGACT Symposium on Theory
  of Computing (STOC)}, 2020.

\end{thebibliography}
\end{document}